\documentclass[aps,pra,twocolumn,twoside,10pt]{revtex4-1}

\usepackage{graphicx,epstopdf,amsmath,amssymb,braket,color,epsfig,epstopdf,geometry,amsthm}
\usepackage[colorlinks=true, citecolor=red, urlcolor=blue]{hyperref}
\usepackage[table]{xcolor}
\usepackage{times}

\newtheorem{theorem}{Theorem}
\newtheorem{corollary}{Corollary}[theorem]
\newtheorem{proposition}{Proposition}
\newtheorem{propcorollary}{Corollary}[proposition]

\begin{document}

\title{Multipartite Entanglement Accumulation in Quantum States:\texorpdfstring{\\}{} Localizable Generalized Geometric Measure}

\author{Debasis Sadhukhan, Sudipto Singha Roy, Amit Kumar Pal, Debraj Rakshit, Aditi Sen(De), and Ujjwal Sen}
\affiliation{Harish-Chandra Research Institute, Chhatnag Road, Jhunsi, Allahabad - 211019, India}
\affiliation{Homi Bhabha National Institute, Training School Complex, Anushaktinagar, Mumbai 400094, India}


\begin{abstract}

Multiparty quantum states are useful for a variety of quantum information and computation protocols. We define a multiparty entanglement 
measure based on local measurements on a multiparty quantum state, and an entanglement measure averaged on the post-measurement 
ensemble. Using the generalized geometric measure as the measure of multipartite entanglement for the ensemble, 
we demonstrate, in the case of several well-known classes of multipartite pure states, that the localized multipartite entanglement 
can exceed the entanglement present in the original state. We also show that measurement over multiple parties may be beneficial in 
enhancing localizable multipartite entanglement. We point out that localizable generalized geometric measure faithfully signals quantum critical phenomena in well-known quantum spin models even when considerable finite-size effect is present in the system. 

\end{abstract}

\maketitle 

\section{Introduction}
\label{intro}

Emergence of multipartite entanglement \cite{horodecki_rmp}
as a crucial ingredient in several information processing tasks like 
measurement-based quantum computation \cite{onewayQC}, 
quantum dense coding \cite{densecode,densecode_bruss,ASDUS}, 
and quantum cryptography \cite{crypto,communication_a,communication_b} has emphasized the importance of
quantifying entanglement in multipartite systems. Multipartite entanglement has been proven essential also in 
detecting cooperative phenomena such as quantum phase 
transitions (QPTs) \cite{geom_ent_manybody,ggmn}, 
and to explain transport properties in photosynthetic 
complexes \cite{transport} (see \cite{trap_rev,fazio_rmp,modi_rmp} for reviews). 
In this respect, it has also been pointed out that multipartite entanglement can be necessary to detect 
some QPTs, which are not clearly signaled by bipartite measures \cite{multi-vs-bi}.  
The growing interest for estimating multipartite entanglement in 
many-body systems is also sustained by impressive experimental 
advances towards creating entangled particles in laboratories with various substrates, {\it e.g.}, trapped ions \cite{ions},  
photons \cite{photons}, superconducting materials \cite{sc-qu}, nuclear magnetic resonance (NMR) \cite{nmr_multi}, and 
optical lattices \cite{optical-lattice} (see also \cite{exp}).
However, despite considerable attempts, progress in developing measures of multipartite entanglement has been limited 
\cite{horodecki_rmp,rel_ent_ent,glob_ent,other_ent,geom_ent,geom_ent_ksep,geom_ent_book,ggm_first,ggmn,score,monogamy,tangle_bunch}.

An interesting utility of multipartite quantum states is the measurement based quantum computation \cite{onewayQC}, where 
quantum gates are implemented by solely performing suitable measurements on different local parts of a previously prepared 
quantum state of a number of parties, eg., on a lattice. In a different situation, one may consider performing measurements on some parts of 
a multiparty quantum state, so that the remaining parties share a useful quantum state. A particularly important example is provided by the 
Greenberger-Horne-Zeilinger (GHZ) state \cite{ghzstate}, given by 
$|\psi\rangle=(|0\rangle^{\otimes N}+|1\rangle^{\otimes N})/\sqrt{2}$. Tracing out $m$ qubits from $|\psi\rangle$, where $m<N-1$,  leads to 
a separable state with vanishing entanglement, while
performing local measurements over the $m$ parties, one obtains an ensemble of pure states conditioned to the measurement 
outcomes, which has non-zero entanglement. Such protocols encourages one to consider a general scenario where an $N$-party quantum state, 
$\rho_N$, is measured at a certain number of parties, to obtain an ensemble of a lower number of parties, with 
the measurements being so chosen that the average entanglement, according to a pre-decided measure, of the post-measurement ensemble is 
maximized. Such entanglement accumulation scheme have been previously used to define entanglement measures \cite{entass,pop_rohr,W-vs-GHZ,loc_ent}, 
where it was usual to choose the number of unmeasured parties as two. Here we go beyond the regime, 
where the number of parties in the post-measurement ensemble is more than two.

In this paper, we introduce a localizable multipartite entanglement (LME) measure, 
in terms of the geometric measures of entanglement \cite{ggm_first,ggmn,score}, and discuss 
its various properties. 
Specifically, we prove that the LME is invariant under local unitary transformations, and show
that it is bounded above by any upper bound of its parent multiparty entanglement measure. 
Note that the concept of LME requires the knowledge of another multiparty entanglement measure, and due to the 
compact computational form of multiparty entanglement, as quantified by the generalized geometric measure (GGM) \cite{ggm_first,ggmn,score}, 
we restrict ourselves to the cases 
in which the GGM is identified with the latter measure. 
We call the corresponding quantity as 
the localizable GGM (LGGM). For arbitrary number of qubits, we analytically 
find the exact expression of LGGM for several classes 
of multipartite states which include  the generalized 
GHZ (gGHZ) state \cite{ghzstate}, generalized W (gW) state \cite{zhgstate,W-vs-GHZ,dvc}, 
and Dicke states with different excitations \cite{dicke,dicke_entangle,dicke_utility}, 
when measurement is restricted to a single qubit. 
Interestingly, we find that in the case of gGHZ state, 
LGGM coincides with the GGM of the original state. On the other hand, 
for gW state, local measurement helps to accumulate higher 
multipartite entanglement in the lower number of qubits, as  compared to the 
content of multipartite entanglement of the original state, 
showing qualitatively distinct behavior than the gGHZ state. Moreover, we prove that the value of LGGM
over two qubits in the case of an arbitrary three-qubit pure state is always 
lower bounded by the value of the geometric measure of the original state, 
while no such bound exists when higher number of 
qubits are involved. 

For specific classes of four- and five-qubit states, we show that local measurement on two parties may help 
to increase LGGM, as compared to the same with only single-qubit measurement.
Extensive numerical simulations seem to imply that such observation holds for almost all four- and 
five-qubit states. We also consider the utility of LGGM in detecting quantum cooperative phenomena in many-body systems. We perform finite-size calculations to show that it can detect the QPTs \cite{qpt_book} occurring in the one-dimensional (1d) quantum Ising model in a transverse field \cite{xy_group}.
Moreover, we show that LGGM signals the QPTs including the Kosterlitz-Thouless (KT) transition of the 1d XXZ model
\cite{xxz_group,mikeska,takahashi,bethe}.

The paper is organized as follows. Sec. \ref{def_prop} contains 
the formal definition of LME using geometric measures as the quantifier. It also 
describes various properties of LME. Sec. \ref{sqm} describes the results 
regarding single-qubit measurement in the case of several well-known 
examples of multiqubit states, such as the generalized GHZ state (gGHZ), 
the generalized W (gW) state, and the $N$-qubit symmetric states. The effect of measurement over 
more than one qubit on the value of LME, along with specific examples in the 
case of four- and five-qubit systems, is discussed in Sec. \ref{mqm}. 
The numerical results regarding arbitrary three-, four-, and five-qubit 
pure states are presented in Sec. \ref{num_res}. 
Sec. \ref{spin_lggm} deals with the 
study of the behavior of LME in well-known quantum spin models. 
Sec. \ref{conclude} contains the concluding remarks.

\section{Localizable Multipartite Entanglement}
\label{def_prop}

In this section, we formally introduce the LME, and discuss its properties. 

\subsection{Definition}
\label{definitions}

The LME of a multiparty pure state can be defined as the maximum average multipartite 
entanglement that can be concentrated over a certain specified set
of parties in the system by performing local measurements over the rest of the parties. 
Let us consider a pure state $|\Phi_{N}\rangle$, which
describes a multipartite system consisting of $N$ parties distinguished by the 
index $i=1,2,\cdots N$. For simplicity, while defining the LME, we consider that the dimension of Hilbert 
space of each of the parties is same, i.e., $d_{i}=d,\,i=1,\cdots,N$. 
However, a more generalized definition using different dimensions for different parties can also be given. 

Let us consider local quantum measurements performed by any $m$ parties on the $N$-party state, $|\Phi_N\rangle$. 
Let $\mathbf{r}$ be the set of positions of the measured qubits, given by  
$\mathbf{r}=\{r_j\},j=1,\cdots,m$, where $r_j\in\{1,2,\cdots,N\}$. The local measurements lead to 
a post-measurement outcome  ensemble, $\{p^l,|\Psi^l_{N}\rangle\}$, of pure states. The index $l$ denotes the running index 
of the measurement outcomes, and runs over the joint Hilbert space of the measured parties having dimension $d^m$.
Here, $p^l$ is the probability of obtaining the state 
$|\Psi^{l}_{N}\rangle=M^l|\Phi_N\rangle$, with $M^l$ being the corresponding measurement operator, 
and $\sum_{l=1}^{d^m}p^l=1$.  
The LME for 
the multipartite state $|\Phi_N\rangle$, for local measurements at $\textbf{r}$, can be defined as 
\begin{eqnarray}
 E_{L,\mathcal{E}}^{n,\mathbf{r}}(|\Phi_N\rangle) = \underset{\mathcal{M}}{\sup}\sum_{l=1}^{d^{m}}p^{l}\mathcal{E}
 \left(|\psi^{l}_{n}\rangle\right),
 \label{lme}
\end{eqnarray}
where $\mathcal{M}\equiv\{M^l\}$ is a set of measurement operators on the $d^m$-dimensional Hilbert space, 
and $\mathcal{E}\left(|\psi\rangle\right)$
is an arbitrary multipartite entanglement measure for the state $|\psi\rangle$. 
Here, $n=N-m$ $(n\leq N-1)$, and is currently redundant in the notation. Its use will become clear once we reach 
Eq. (\ref{glob_lggm}). We will also assume that $n\geq 2$, although the formalism adopted here can be generalized to the 
case of $n=1$ also. 
The state $|\psi^l_n\rangle$ corresponding to the $N$-party
state $|\Psi^l_{N}\rangle$, for a fixed $\mathbf{r}$ and a fixed measurement outcome $l$, is obtained by tracing out 
the $m$ parties of which local measurements are performed. The tracing out operation is performed after the local measurement has been carried out.
The supremum is taken over all complete sets of measurement operators in
$\mathcal{M}$, since the supremum is not guaranteed to be attained within the set, due to the possible complex nature of $\mathcal{M}$.   

It is important to note here that the conceptualization of an LME depends on the understanding of another measure of multiparty entanglement 
of a lower number of parties. This latter measure is in some sense acting as a ``seed measure'' for the LME. 
To define LME, one may consider measurement protocols corresponding 
to, for example, projective measurements (PV) without classical communication between the parties, or positive operator valued 
measures (POVMs) without classical communication between the parties, or general local operations and classical communication (LOCC). 
In the last case, the classical communication (CC) is 
among the $m$ parties over which the local operations are performed. 
It is clear that 
\begin{eqnarray}
 E_{L,\mathcal{E}}^{n,\mathbf{r}}|_{\mbox{\tiny PV\normalsize}}
 \leq E_{L,\mathcal{E}}^{n,\mathbf{r}}|_{\mbox{\tiny POVM\normalsize}}
\leq E_{L,\mathcal{E}}^{n,\mathbf{r}}|_{\mbox{\tiny LOCC\normalsize}},
\label{relative}
\end{eqnarray}
for a fixed multipartite state, $|\Phi_{N}\rangle$, a fixed set of 
measured parties, $\mathbf{r}$, and a 
chosen multiparty entanglement measure, $\mathcal{E}$.  

One must note here that for a fixed initial multipartite state 
$|\Phi_N\rangle$, and a fixed set of measurement protocols, the value of LME depends on two factors: \textbf{(i)} 
the multipartite entanglement measure $\mathcal{E}$ (the seed measure), and \textbf{(ii)} the set $\mathbf{r}$, i.e., 
the choice of $m$ parties over which local measurements are performed. Before discussing the choice of $\mathcal{E}$, 
let us briefly consider the dependence of LME over the set $\mathbf{r}$ that is inherent in the definition.
For an $N$-partite system with local measurements at $m$ parties, a set $\mathbf{R}=\{\mathbf{r}_\alpha\}$,
$\alpha=1,2,..,\binom{N}{m}$, of all possible choices of $\mathbf{r}$ exists, 
thereby allowing $\binom{N}{m}$ number of values of LME, 
$E_{L,\mathcal{E}}^{n,\mathbf{r}_\alpha}$.  
Therefore, the LME, $E^{n,\mathbf{r}_\alpha}_{L,\mathcal{E}}$, 
is ``local'' in the sense that it changes with the 
choice of the set $\mathbf{r}_{\alpha}$. In light of this fact, one can also define a ``global'' value of the LME for a 
multipartite state with fixed values of $N$ and $m$ as 
\begin{eqnarray}
 E_{GL,\mathcal{E}}^{n}=\underset{\mathbf{R}}{\max}\{E_{L,\mathcal{E}}^{n,\mathbf{r}_\alpha}\}.
 \label{glob_lggm}
\end{eqnarray}
Note that if the initial state $|\Phi_N\rangle$ is a symmetric state, such maximization over $\mathbf{R}$ is not required. 
Note also that the relative positions of the parties labeled by ``$\mathbf{r}_\alpha$'', for a specific $\alpha$, 
with respect to each other as well as the rest of the parties does not affect the value of LME. 
From now on, without any loss of generality, we assume that the measurements are performed at 
$r_j=j$, $j=1,2,\cdots,m$, while $|\psi^l_n\rangle$ denotes the state of the 
remaining $n$ parties.

This definition can be extended to an arbitrary mixed state when the 
input state is an $N$-party state $\varrho_{1,2,\cdots,N}$, 
and the measurement is performed on any $m$ parties. The problem in this 
case remains with the choice of a computable multipartite  entanglement measure $\mathcal{E}$, which 
is defined for the mixed states. 
Although the notion of entanglement measures in multipartite systems 
is an active field of research, the number of such 
computable measures, even in the case of the pure states, is still limited. Another avenue to extend the LME to mixed states is via the convex-roof optimization. Convex-roof optimization is, however, typically difficult to perform, and has been successful in only a few instances. See \cite{eof} for examples within quantum information.

Reverting back to the case of pure states, in principle, 
an LME can be defined by using any one of the known candidates
for pure state multipartite entanglement measure as the seed measure, such as relative entropy of 
entanglement \cite{rel_ent_ent}, 
global entanglement \cite{glob_ent}, and other 
multipartite measures \cite{horodecki_rmp,other_ent}. 
However, in the present study, we focus on geometric measures (GM) of entanglement 
\cite{horodecki_rmp,geom_ent,geom_ent_manybody,geom_ent_ksep,geom_ent_book,ggm_first,ggmn} for multipartite pure states.
More specifically, we use the ``$K$-separability based GM'' (K-GM) to discuss the properties of the LME. 
For the purpose of computation, we choose the generalized geometric measure (GGM), which is 
computable for a multiparty pure state in arbitrary dimensions and for arbitrary number of parties. Short descriptions of 
these measures can be found in Appendix \ref{mult-mes}.

In this paper, we have confined ourselves to consideration of measurements that are local. It is possible to define a multipartite entanglement measure for which non-local measurements are allowed in the optimization. However, non-local measurements, even if performed on two parties at a time, can generate genuine multiparty entanglement, and therefore will result in difficulties in defining the monotonicity of the so-obtained measure under LOCC. Furthermore, parametrization of the non-local measurements via entangled bases require a large number of parameters that increase exponentially with the increase in the number of parties measured.

\subsection{Properties}
\label{properties}

We now prove several properties of LME. 
We use K-GM as the multipartite 
entanglement measure, and local projective measurements, in which case, Eq. (\ref{lme}) reads
\begin{eqnarray}
 E_{L}^{n,\mathbf{r}_\alpha}(|\Phi_N\rangle) = \underset{\mathcal{P}}{\sup}\sum_{l=1}^{d^{m}}p^{l}G_K
 \left(|\psi^{l}_{n}\rangle\right),
 \label{lgm}
\end{eqnarray}
where $\mathcal{P}\equiv\{P^l\}$ denotes a complete set of local projectors 
acting on the parties distinguished by $\mathbf{r}_\alpha$. 
From now onward, we discard the index $\mathcal{E}$. 
We start by looking into the bounds of the measure, which leads us to the following theorem. 

\begin{theorem}
For an arbitrary state $|\Phi_N\rangle$ describing a quantum system of $N$ parties, 
$0\leq E_L^{n,\mathbf{r}_{\alpha}}\leq g$, where $G_K\leq g$.
\label{th:bounds}
\end{theorem}

\noindent Theorem \ref{th:bounds} provides an upper bound of LME depending on the choice of the seed measure, provided an upper bound is known for the seed. Corollary \ref{cor:gm_bound} follows directly from Theorem \ref{th:bounds}, and the definition of $E^n_{GL}$ (Eq. (\ref{glob_lggm})).

\begin{corollary}
For an arbitrary state $|\Phi_N\rangle$ describing a quantum system of $N$ parties, 
$0\leq E^n_{GL}\leq g$, where $G_K\leq g$.
\label{cor:gm_bound} 
\end{corollary}

\noindent Note that for $K=2$, $G_K\equiv\mathcal{G}$ for an $N$-qubit system, and the above property implies  
$E^{n,\mathbf{r}_\alpha}_L\leq 1/2$, since $\mathcal{G}\leq 1/2$.  
Our next theorem is on the effect of local unitary (LU) operations on LME.

\begin{theorem}
$E_L^{n,\mathbf{r}_\alpha}$ remains invariant under local unitary transformations.
\label{th:unitary}
\end{theorem}

\noindent To determine the criteria for vanishing $E^{n,\mathbf{r}_\alpha}_L$, we intend to 
characterize the set of multipartite states, 
$\{|\Phi_N\rangle\}$, for which $E_{L}^{n,\mathbf{r}_\alpha}=0$. We first consider the following theorem.

\begin{theorem}
For a $K$-separability based GM, a non-zero value of $E^{n,\mathbf{r}_\alpha}_L$
for an $N$-partite pure state, $|\Phi_N\rangle$, 
is possible with $n=N-m$, $m$ being the number of parties in which measurement is performed, only when
the separability, $M$, of the original state is such that $1\leq M\leq K+m-1$.
\label{th:conditions} 
\end{theorem}

\begin{proof}
To prove the above theorem, let us first assume that the 
$N$-partite system is composed of two partitions denoted
by $A$, and $B$. The first one consists of the $m$ parties over 
which local measurements are performed, while the rest $n=N-m$ parties
construct the partition $B$. For the value of $E_L^{n,\mathbf{r}_\alpha}$ 
to be non-zero, at least one of the states, $\{\psi^l_n\}$, of 
the post-measurement ensemble consisting of $d^m$ states must not be $K$-separable, so that $G_K(|\psi^l_n\rangle)\neq0$. 
This implies that the state $|\psi^l_{n}\rangle$ is allowed to be $K^\prime$-separable, where $1\leq K^\prime\leq K-1$. 
Moreover, irrespective of whether the partitions $A$ and $B$ share entanglement 
among each other, the possible separability of the 
partition $A$ consisting of $m$ parties dictates that the original state, $|\Phi_N\rangle$, is allowed to be 
$(K^\prime+j)$-separable, where $1\leq j\leq m$. Combining these two results, one obtains the allowed range of $M$ 
as $1$ to $K+m-1$. 
\end{proof}

\noindent All the characteristics of LME discussed above remains unchanged if 
K-GM is replaced by GGM as the genuine multipartite entanglement measure. 

\section{LGGM in multiparty quantum states}

\subsection{Single-qubit measurement: LGGM vs. GGM}
\label{sqm}

Let us consider an $N$-party pure state in $(\mathbb{C}^2)^{\otimes N}$. 
The local projective measurement, $\mathcal{P}$, on a qubit $j$ ($j=1,2,\cdots,N$), 
can be represented by a complete set of rank-$1$ projectors, $\{\Pi^l_j\}$, such that 
$\Pi^l_j=|\xi^l_j\rangle\langle\xi^l_j|\otimes\mathbb{I}_{N-1}$, $l=1,2$, is given by 
\begin{eqnarray}
 |\xi^{1}_j\rangle&=& c_{\theta_j/2}|0\rangle+e^{i\phi_j}s_{\theta_j/2}|1\rangle,\nonumber \\
 |\xi^{2}_j\rangle&=& -s_{\theta_j/2}e^{-i\phi_j}|0\rangle+c_{\theta_j/2}|1\rangle,
 \label{projector}
\end{eqnarray}
with $0\leq\theta_j\leq\pi$, $0\leq\phi_j<2\pi$, and $c_x$ and $s_x$ stand for $\cos x$ and $\sin x$ respectively. 
Here, $\mathbb{I}_{N-1}$ is the identity operator in the Hilbert space of the $N-1$ qubits, 
and $\{|0\rangle,|1\rangle\}$ is the computational basis in the qubit 
Hilbert space. 
In this representation, the supremum involved in the definition of LME is obtained by
performing a maximization over the space of the real parameters, $(\theta_j,\phi_j)$. 
If local measurements are performed over a collection of 
qubits denoted by $\mathbf{r}\equiv\{r_j\}$, $j=1,2,\cdots,m$ with
$m>1$, the supremum has to be obtained over a total of $2m$ real parameters, $(\theta_{r_j},\phi_{r_j})$. 
Using GGM as the seed measure, Eq. (\ref{lme}) takes up a simpler form, given by
\begin{eqnarray}
E_{L}^{r} = \underset{\mathcal{P}}{\mbox{sup}}\sum_{l=1}^{2}p^{l}\mathcal{G}(|\psi^l_{N-1}\rangle).
\label{lmen}
\end{eqnarray}
Note that in Eq. (\ref{lmen}), we have discarded the superscript $n$ since it has a constant value $n=N-1$ in the 
present case. To keep the notations uncluttered, we also replace $\mathbf{r}_\alpha$ with the position index, $r$, since 
$r$ can now have $N$ possible values, i.e.,  $r=1,2,\cdots,N$. In this section, and in the rest of the paper, 
unless otherwise stated, we always consider local measurement over the first qubit of the system in the case of a single-qubit measurement.


\noindent\textbf{\textit{Generalized GHZ state.}} The first example that we consider is the $N$-qubit gGHZ state,  given by \cite{ghzstate}
$|GHZ_{N}\rangle_g=a_1|0\rangle^{\otimes N}+a_2|1\rangle^{\otimes N}$,
where $a_1$ and $a_2$ are complex numbers with $|a_1|^2+|a_2|^2=1$. Without any loss of generality, let us assume that
$|a_1|^2\geq\frac{1}{2}\geq|a_2|^2$. Since $|GHZ_N\rangle_g$ is symmetric under swapping of parties, $E_{L}^{r}=E_L$ for $r=1,2,\cdots,N$. 
The following proposition is for the LGGM of $N$-qubit gGHZ states (see Appendix \ref{prop_proof} for the proof).
\begin{proposition}
For the $N$-qubit gGHZ state, $E^1_L=\mathcal{G}$.
\label{pr:ghz}
\end{proposition}
\noindent The value of LGGM remains unchanged in the case of the gGHZ state if measurement is performed over a higher number of qubits 
$(m>1)$. However, in subsequent discussions, we shall be providing examples of multipartite quantum states for which the 
situation is different. As a special case of the $N$-qubit gGHZ state, the LGGM for the GHZ state 
$(a_1=a_2=1/\sqrt{2})$ of $N$ qubits can be obtained as $E_L^1=1/2$.

\begin{figure}
 \includegraphics[scale=0.325]{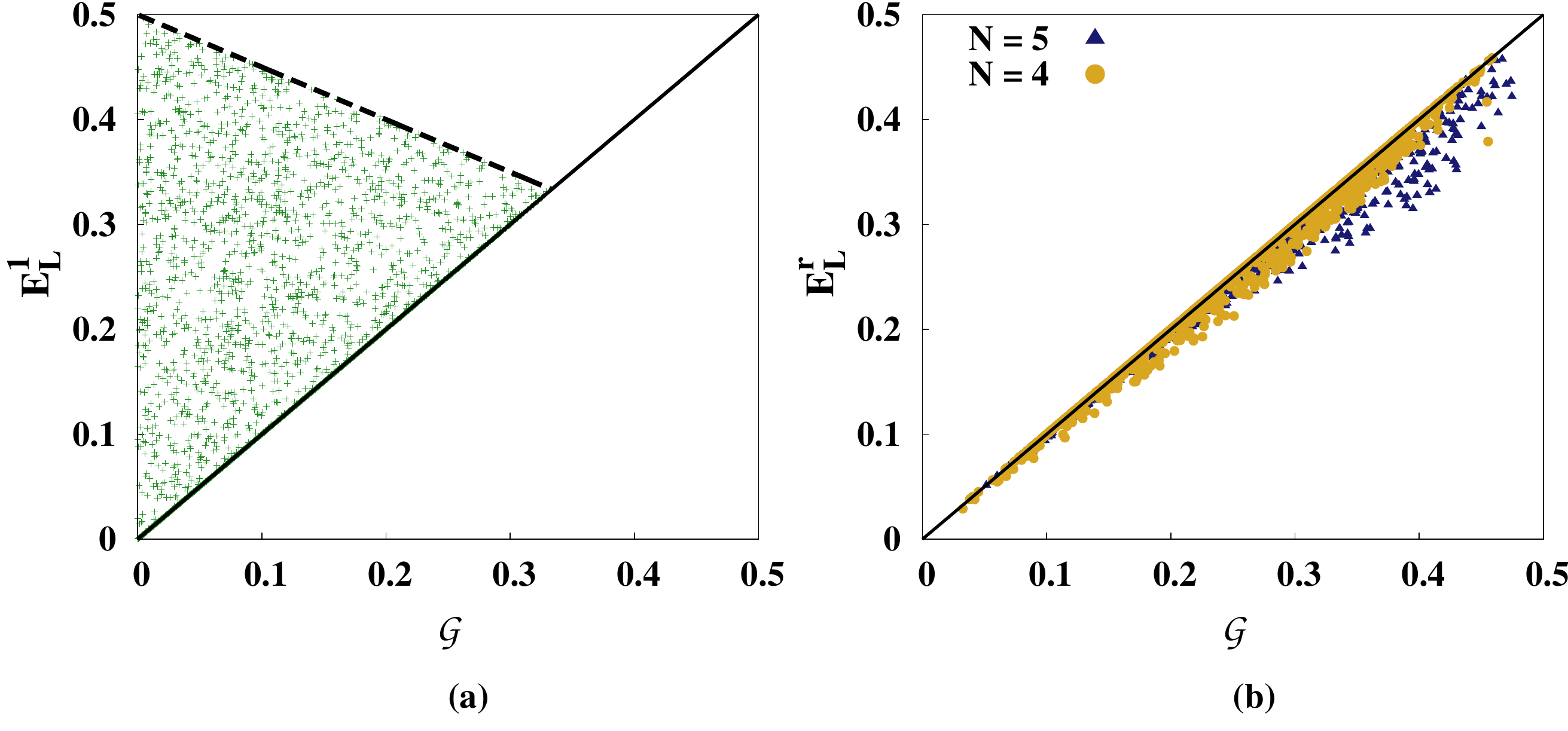}
 \caption{(Color online.) (a) Plot of $E_{L}^1$ vs. $\mathcal{G}$ for three-qubit gW states. 
 To obtain the scatter diagram, $10^5$ three-qubit gW states are generated Haar-uniformly. 
 The solid line represents the line $E_L^1=\mathcal{G}$, while the dashed line correspond to 
 $2E_L^1+\mathcal{G}=1$. (b) Plot of $E_{L}^1$ vs. $\mathcal{G}$ in the case of Haar-uniformly generated 
 generalized superposition of Dicke states, as given in 
 Eq. (\ref{gen_sup_dicke}). To obtain the scatter diagram, $10^5$ states of the 
 form $|D_g^N\rangle$ are generated Haar-uniformly  
 for each of the cases $N=4$ and $N=5$.
 All quantities plotted in both figures are dimensionless.}
 \label{genw_sym}
\end{figure}

\begin{table}
 \begin{tabular}{|c|c|c|c|c|c|c|}
 \hline
 \cellcolor{blue!10} \#\, & \cellcolor{green!10} \textbf{Ordering} 
 & \cellcolor{red!10} $\mathbf{E_{L}^1}$ &  \cellcolor{red!10} $\mathbf{E_{L}^2}$ & 
 \cellcolor{red!10} $\mathbf{E_{L}^3}$ & \cellcolor{cyan!10} 
 $\mathbf{\mathcal{G}}$ & \cellcolor{yellow!10} \textbf{Partition}  \\
 \hline 
 $1$ & $|a_1|^2\geq|a_2|^2\geq|a_3|^2$ & $|a_2|^2$ & $|a_3|^2$ & $|a_3|^2$ & $|a_3|^2$ & 1:23 \\
 \hline 
 $2$ & $|a_1|^2\geq|a_3|^2\geq|a_2|^2$ & $|a_2|^2$ & $|a_3|^2$ & $|a_2|^2$ & $|a_2|^2$ & 2:13 \\
 \hline
 $3$ & $|a_2|^2\geq|a_1|^2\geq|a_3|^2$ & $|a_1|^2$ & $|a_3|^2$ & $|a_3|^2$ & $|a_3|^2$ & 1:23 \\
 \hline
 $4$ & $|a_2|^2\geq|a_3|^2\geq|a_1|^2$ & $|a_1|^2$ & $|a_1|^2$ & $|a_3|^2$ & $|a_1|^2$ & 3:12 \\
 \hline
 $5$ & $|a_3|^2\geq|a_1|^2\geq|a_2|^2$ & $|a_2|^2$ & $|a_1|^2$ & $|a_2|^2$ & $|a_2|^2$ & 2:13 \\
 \hline
 $6$ & $|a_3|^2\geq|a_2|^2\geq|a_1|^2$ & $|a_1|^2$ & $|a_1|^2$ & $|a_2|^2$ & $|a_1|^2$ & 2:13 \\ 
 \hline
 \end{tabular}
\caption{Different orderings of $\{|a_1|^2,|a_2|^2,|a_3|^2\}$ and corresponding values of $E_L^r$ ($r=1,2,3$), 
and $\mathcal{G}$ in the case of three-qubit gW state. For all orderings, 
$E_L^1\geq\mathcal{G}$. The last column shows the bipartition 
from which the maximum Schmidt coefficient is obtained.}
\label{tab1}
\end{table}

\noindent\textbf{\textit{Generalized W state.}} Our next example is the $N$-qubit gW state, given by \cite{zhgstate,W-vs-GHZ,dvc}
$|W_N\rangle_g=\sum_{i=1}^{N}a_i|0\rangle^{\otimes(i-1)}|1\rangle_i|0\rangle^{\otimes(N-i)}$,
where $\{a_i\}$, $i=1,2,\cdots,N$, are complex numbers such that $\sum_{i=1}^{N}|a_i|^2=1$. 
Note that unlike the gGHZ state, the gW state is not symmetric under swapping of parties, which leads to a collection of 
$N$ values of LGGM, $\{E_{L}^r\}$, $r=1,2,\cdots,N$. The GGM of the state, in the present case, is given by 
$\mathcal{G}=\min\{|a_i|^2\}$, where $i=1,2,3$. For the purpose of demonstration, we start 
with the case of $N=3$, for which the relation between $E_{L}^r$ and $\mathcal{G}$ is given by the following proposition 
(see Appendix \ref{prop_proof} for the proof).
\begin{proposition}
For an arbitrary three-qubit gW state, $E_{L}^r\geq \mathcal{G}$ for all values of $r$.
\label{pr:3gw}
\end{proposition}
\noindent The following corollary regarding the lower bound of the ``global'' LGGM, as defined in Eq. (\ref{glob_lggm}), 
can be obtained directly from Proposition \ref{pr:3gw}.
\begin{propcorollary}
For the tripartite gW state, the global LGGM, $E_{GL}\geq\mathcal{G}$.
\label{pc:lbound}
\end{propcorollary}
\noindent In the case of $N$-qubit gW states with $N>3$, the situation is more involved. The difficulty in 
determining the values of GGM of the states $|\psi^l\rangle$ in the post-measurement ensemble, which now correspond to 
$N-1$ qubits or less, makes  analytical optimization of LGGM rather difficult. 
However, motivated by the Proposition \ref{pr:3gw}, we intend to check whether such a lower bound of LGGM exists if 
one considers an $N$-qubit gW state, when measurement is 
performed only on a single qubit. This leads us to Proposition \ref{pr:gw_ubound}. The proof of the Proposition is given in 
Appendix \ref{prop_proof}.
\begin{proposition}
For an arbitrary gW state of $N$-qubits, $E_{L}^r\geq\mathcal{G}$ for 
all values of $r$, provided measurement is performed only on a single qubit.
\label{pr:gw_ubound}
\end{proposition}
\noindent Our numerical analysis suggests that irrespective of the value of $N$, the maximization involved in LGGM 
for the gW states, is obtained when the local 
projective measurement is performed in the computational basis, $\{|0\rangle,|1\rangle\}$.  
Using this information, an upper bound for $E^r_L$ of the state $|W_N\rangle_g$ can also be proved,
as given in Proposition \ref{pr:spgw} (see Appendix \ref{prop_proof} for the proof).
\begin{proposition}
For an arbitrary $N$-qubit gW state, 
$E_{L}^r$, for all values of $r$, is bounded above by the LGGM of the gW state having the same value of $\mathcal{G}$, 
but with squared modulus of all the coefficients except one, denoted by $|a_i|^2$, being equal to 
$|a_{j}|^2=(1-|a_i|^2)/(N-1)$, where $|a_i|^2=\min\{|a_k|^2\}$, $k=1,2,\cdots,N$, $j\neq i$, and 
$i,j\in\{1,2,\cdots,N\}$.
\label{pr:spgw}
\end{proposition}

Fig. \ref{genw_sym}(a) depicts the scatter diagram of $E_L^1$ vs. $\mathcal{G}$ in the case of a set of $10^5$ Haar-uniformly 
generated three-qubit gW states. The solid line, representing $E_L^1=\mathcal{G}$, depicts the lower bound of $E_L^1$ 
(Proposition II), while the dashed line, representing $2E_L^1+\mathcal{G}=1$, correspond to the upper bound given in 
Proposition \ref{pr:spgw}. Note that all the points in the scatter diagram are enclosed by these two lines, and $\mathcal{G}=0$.
The W state, $|W\rangle$, is obtained from the gW state with $a_i=1/\sqrt{N}$, $i=1,2,\cdots,N$. Evidently, the GGM and 
LGGM of the W state are equal to each other, having a value $1/N$.

Considering that LGGM of the $N$-qubit gW state is maximized when measurement is performed in $\{|0\rangle,|1\rangle\}$, we make the 
following observation. 

\noindent\textbf{Observation.} For an $N$-qubit gW state, $E^r_L\geq\mathcal{G}$ if 
the value of $\mathcal{G}$ is obtained from $r$:rest bipartition, while in all other cases, 
$E^r_L=\mathcal{G}$.

\noindent We demonstrate the implications of the above observation in Table \ref{tab1}, where 
the $E^r_L$ of three-qubit gW states in different cases (as distinguished by the different orderings of the $|a_i|^2$), 
are tabulated along with their GGMs. 
Note that the global LGGM, $E_{GL}$, is obtained when local projective measurement is performed on 
the qubit $r$ such that the maximum Schmidt coefficient is obtained from the bipartition $r$:rest, 
as presented in the last column. From the table, it is clear that 
the value of $E_{L}^{r^\prime}$, when the $r^\prime$:rest bipartition does not provide the maximal Schmidt 
coefficient, is always equal to the GGM of the three-qubit gW state, which is in agreement with the observation. 
We shall be investigating this issue, and its possible generalization, while discussing the LGGM of arbitrary 
three-qubit pure states, in Sec. \ref{num_res}.

\noindent\textbf{\textit{Symmetric States.}} We next consider a special class of multipartite quantum states known as the ``symmetric
states'', which remains unaltered with the permutation of parties. First, we investigate the behavior of LGGM in a special subset 
of symmetric states -- the highly entangled $N$-qubit Dicke states \cite{dicke,dicke_entangle} --
with high applicational advantages \cite{dicke_utility}. An $N$-qubit Dicke state with $k$ excitations can be represented as 
\begin{eqnarray}
 |D^{N}_{k}\rangle
 =\frac{1}{\sqrt{\binom{N}{k}}}\sum_{i}\mathcal{P}_i\left(|0\rangle^{\otimes N-k}\otimes|1\rangle^{\otimes k}\right),
 \label{dickestate}
\end{eqnarray}
where the summation is over all possible permutations of $N$-qubit product states composed of $N-k$ qubits in the 
ground state, $|0\rangle$, and the 
rest $k$ qubits in the excited state, $|1\rangle$. 
Note that the $N$-qubit W state can be identified as $|D^N_1\rangle$. 
The GGM of $|D_k^N\rangle$, with $N>2$,
can be obtained as \cite{dicke_entangle}
\begin{eqnarray}
  \mathcal{G}=
    \begin{cases}
     \frac{N-2}{2(N-1)} \mbox{ for } k=\frac{N}{2},\\
     \frac{k}{N} \mbox{ for } k<\frac{N}{2}.
   \end{cases}
   \label{dicke_ggm}
\end{eqnarray}
To determine $E_{L}^r$, $r=1,\cdots,N$, of the Dicke state, one has to perform a 
rank-$1$ projective measurement on any one of the qubits. The details on the post-measurement ensemble and 
determination of GGM can be found in Appendix \ref{ap:dicke}.
We perform extensive numerical analysis for different values of $N$ and $k$ to 
find that irrespective of the number of qubits 
and the number of excitations, the optimization of LGGM for $|D^N_k\rangle$ always takes place 
at $\theta=0,\phi=0$, i.e., at the $\{|0\rangle,|1\rangle\}$ basis. The LGGM, in that case, can be obtained as explicit 
functions of $N$ and $k$, as 
\begin{eqnarray}
 E_L^r(\mbox{even }N)&=&
    \begin{cases}
     \frac{N-2}{2(N-1)} \mbox{ for } k=\frac{N}{2},\\
     \frac{k}{N} \mbox{ for } k<\frac{N}{2}. 
    \end{cases}\nonumber \\
 E_L^r(\mbox{odd }N)&=&
    \begin{cases}
     \frac{k}{N} \mbox{ for } k<\frac{N-1}{2},\\
     \frac{N-1}{2N}-\frac{N+1}{4N(N-2)} \mbox{ for } k=\frac{N-1}{2},N>3,\\
     \frac{1}{3} \mbox{ for } k=1,2 \mbox{ and } N=3.
   \end{cases} \nonumber \\
   \label{dicke_lggm}
\end{eqnarray}
\noindent From Eqs. (\ref{dicke_ggm}) and 
(\ref{dicke_lggm}), it is evident that for $N>3$, $E^r_L<\mathcal{G}$ for odd $N$ with $k=(N\pm1)/2$,
while in all other cases, $E_L^r=\mathcal{G}$, thereby suggesting an upper 
bound of LGGM, given by $E_{L}^r\leq\mathcal{G}$ for 
symmetric Dicke states of $N$ qubits with $k$ excitations provided the maximization is obtained in the 
computational basis.

\begin{figure}
 \includegraphics[scale=0.325]{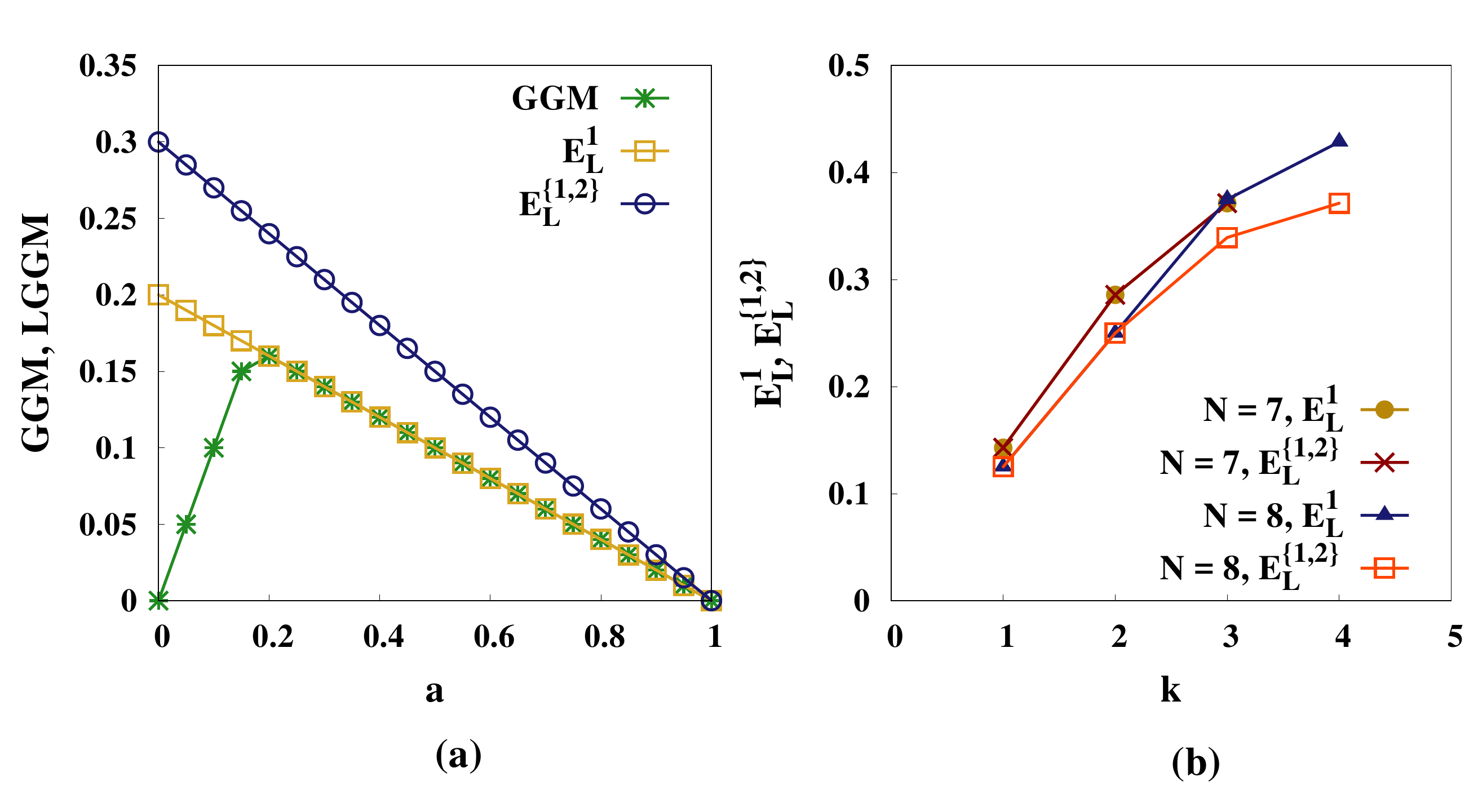}
 \caption{(Color online.) \textbf{(a)} Plot of $\mathcal{G}$, $E_L^1$, and $E_L^{\{1,2\}}$ 
 as functions of $a$ in case of the four-qubit gW state 
 $|W^4\rangle_g$ with  $a_1=a$, $a_2=\sqrt{(1-a^2)/5}$, $a_3=\sqrt{3(1-a^2)/10}$, and
 $a_4=\sqrt{(1-a^2)/2}$. 
 \textbf{(b)} Plot of $E_L^1$ and $E_L^{\{1,2\}}$ with $k$ for $N=7$ and $N=8$ in the case of Dicke states.
 All quantities plotted are dimensionless.}
 \label{mm}
\end{figure}

%

A general form of $N$-qubit symmetric states can be constructed 
by making superposition of all the Dicke states. 
Such a state can be written as 
\begin{eqnarray}
 |D_g^N\rangle=\sum_{k=0}^{N}a_{k}|D^N_k\rangle,
 \label{gen_sup_dicke}
\end{eqnarray}
with complex $\{a_k\}$ such that $\sum_{k=0}^N|a_{k}|^2=1$. Here, the summation index, $k$, 
denotes the possible number of excitations in a Dicke state, 
and can have values $0\leq k\leq N$. 
The difficulty in computing the GGM of an $N$-qubit 
symmetric state of the form given in Eq. (\ref{gen_sup_dicke}) restricts one 
to calculate the LGGM of 
such states only numerically. 
Our numerical 
analysis suggests that for $N=3$, $E_L^1=\mathcal{G}$, while
the upper bound, $E^r_L\leq\mathcal{G}$ holds in the case $N=4,5$, 
like in the case of 
Dicke states.
Fig. \ref{genw_sym}(b) depicts the variation of $E^r_L$ against 
$\mathcal{G}$ for $10^5$ randomly chosen symmetric states  
for each of the cases $N=4$, and $5$, 
where the upper bound is satisfied. 
For $N=4$, and $5$, the percentage of states for which $E^r_L<\mathcal{G}$ are $33.4\%$, and $46.8\%$, respectively,
while for the rest of the states, $E_L^r=\mathcal{G}$, up to four decimal places.

\begin{table*}
 \begin{tabular}{|c|c|c|c|c|c|c|c|c|c|c|c|}
 \hline
 \cellcolor{blue!10} \textbf{State} 
 &\cellcolor{cyan!10} $\mathbf{\mathcal{G}}$ 
 &\cellcolor{red!10} $\mathbf{E_{L}^1}$ 
 &\cellcolor{red!10} $\mathbf{E_{L}^2}$ 
 &\cellcolor{red!10} $\mathbf{E_{L}^3}$ 
 &\cellcolor{red!10} $\mathbf{E_{L}^4}$   
 &\cellcolor{yellow!10} $\mathbf{E_{L}^{\{1,2\}}}$ 
 &\cellcolor{yellow!10} $\mathbf{E_{L}^{\{1,3\}}}$
 &\cellcolor{yellow!10} $\mathbf{E_{L}^{\{1,4\}}}$
 &\cellcolor{yellow!10} $\mathbf{E_{L}^{\{2,3\}}}$
 &\cellcolor{yellow!10} $\mathbf{E_{L}^{\{2,4\}}}$
 &\cellcolor{yellow!10} $\mathbf{E_{L}^{\{3,4\}}}$
 \\
 \hline 
  $|\Psi^4_7\rangle$ & $1/4$ & $1/4$ & $1/4$ & $1/4$ & $1/4$ & $1/4$ & $1/4$ & $1/4$ & $1/4$ & $1/4$ & $1/4$\\
 \hline 
  $|\Psi^4_8\rangle$ & $1/4$ & $1/2$ & $1/4$ & $1/4$ & $1/4$ & $1/2$ & $1/2$ & $1/2$ & $1/2$ & $1/2$ & $1/4$\\
 \hline
 $|\Psi^4_9\rangle$ & $0$ & $1/2$ & $0$ & $0$ & $0$ & $1/2$ & $1/2$ & $1/2$ & $1/2$ & $1/2$ & $0$\\ 
 \hline
 \end{tabular}
\caption{The values of $\mathcal{G}$, $E_L^{r}$ ($r=1,2,3,4$), and 
$E_L^{\{r_1,r_2\}}$ ($r_1\neq r_2$, $r_{1,2}\in\{1,2,3,4\}$) 
for the four-qubit states $|\Psi^4_7\rangle$, $|\Psi^4_8\rangle$, and $|\Psi^4_9\rangle$. }
\label{tab2}
\end{table*}

\subsection{Can local measurement over more than one qubit be beneficial?}
\label{mqm}
In this section, we focus on the question as to whether increasing the number of measured parties can help to increase localizable 
entanglement, and discuss examples of multipartite entangled states in this context. To begin with, we point out that in the case of 
$N$-qubit gGHZ state, $E_L^r=E_{L}^{\mathbf{r}_\alpha}$, where $\mathbf{r}_{\alpha}=\{r_j\}$, $j=1,2,\cdots,m$, 
and $1< m\leq N-2$. This implies that LGGM does not depend on the number of measured qubits for the gGHZ state. 
The situation, however, can be  drastically different if one considers the gW state of $N$ qubits. 
As an example, we consider the case of a four-qubit gW state, where the coefficients $\{a_i\}$, $i=1,\cdots,4$, are real, such that 
$a_1=a$, $a_2=\sqrt{(1-a^2)/5}$, $a_3=\sqrt{3(1-a^2)/10}$, and 
$a_4=\sqrt{(1-a^2)/2}$. For this state, $\mathcal{G}<E_L^1<E_L^{\{1,2\}}$
with $a\leq0.17$, while $\mathcal{G}=E_L^1<E_L^{\{1,2\}}$ when $a>0.17$ (Fig. \ref{mm}(a)), thereby indicating an advantage of measuring 
more than one qubit.  
Such instances can also be found in four-qubit states other than gW states, and in quantum states involving higher number of parties 
(for such examples, see Appendix \ref{ap:example}).
On the contrary, there exist quantum states, for example, $N$-qubit Dicke states, for which $E_{L}^{\{1,2\}}$ is not beneficial 
compared to $E_L^1$. We determine the values of $E_{L}^1$ and $E_{L}^{\{1,2\}}$ in the case of 
$|D^N_k\rangle$ with $4\leq N\leq10$, for different allowed values of $k$, and find that irrespective of the 
value of $N$ and $k$, $E_{L}^{\{1,2\}}\leq E_L^1\leq\mathcal{G}$. 
The variations of $\mathcal{G}$, $E_{L}^1$, and $E_{L}^{\{1,2\}}$
against different values of $k$ for $N=7$ and $8$ are shown in Fig. \ref{mm}(b), while the variations of $E_L^1$ and $E_L^{\{1,2\}}$ 
with varying $k$, for $N=9$ and $10$, are qualitatively similar to that for $N=7$ and $8$ respectively. 
The question of whether this is a generic property of the symmetric states is discussed in the next section. 
One must note that there exists multiparty states other than Dicke states, for which local measurement over more than one qubit may not 
be advantageous, as discussed in the next section. 


\subsection{Arbitrary N-qubit pure states: Numerical Results}
\label{num_res}

We now consider arbitrary $N$-qubit pure states, and compare advantages of multi-qubit measurements over single-qubit ones via 
numerical analysis. We focus on $N$-qubit pure states with $N=3,4$ and $5$, and consider different classes of such states.
Unless otherwise stated, we Haar-uniformly generate $10^5$ arbitrary pure states in each case, and compute
$\mathcal{G}$ for the original state, $E_{L}^1$, and $E_L^{\{1,2\}}$ in the case of each $N$-qubit state. In our numerical analysis, 
the values of two quantities are considered to be equal when they are same upto four decimal places. 

In the three-qubit scenario, we separately consider arbitrary three-qubit pure states belonging to the paradigmatic GHZ and the W classes 
\cite{dvc} (see Appendix \ref{ap:ghzw} for a short description of the classes, and the post-measurement ensembles), 
which are mutually disjoint sets that together construct the entire set of three-qubit pure states. 
We find that for an arbitrary three-qubit pure state, $E_L^1\geq\mathcal{G}$, which is evident from Fig. \ref{fourq}(a).
We perform an extensive numerical search over a set of $10^7$ three-qubit pure states from each of the 
W and the GHZ classes, and find that for all instances, $E_{GL}=E_L^{r^\prime}$, when the bipartition $r^\prime$:rest 
provides $\mathcal{G}$ -- an observation similar to what is found in the case of single-qubit measurement on an $N$-qubit gW state. 
When $r\neq r^\prime$, $E_{L}^r=\mathcal{G}$. 
These findings lead us to the following conjecture:

\begin{itemize}
 \item \textit{If the maximal Schmidt coefficient for an arbitrary three-qubit pure state is obtained across 
the bipartition $r$:rest, $r\in\{1,2,3\}$, then $E_{GL}=E_{L}^{r^\prime}\geq\mathcal{G}$ with $r^\prime\in\{1,2,3\}$ when $r^\prime=r$. 
On the other hand, $E_L^{r^\prime}$ coincides with $\mathcal{G}$ when $r^\prime\neq r$.}
\end{itemize}
\noindent This conjecture helps to pin-point the position of measurement while one tries to increase the value of genuine multiparty 
entanglement by means of localization. 

%

To investigate arbitrary four-qubit pure states, we focus on the classes given by $|\Psi^4_i\rangle$, $i=1,2,\cdots,6$, from 
the nine parametric classes of four-qubit states \cite{slocc,infinite} (see Appendix \ref{ap:fourq}). 
Note that in all instances except $|\Psi^4_6\rangle$, the LGGM is found to be upper bounded by the GGM of the original state
(see Fig. \ref{fourq}(b)), when measurement is performed on qubit $1$. 
This result remains unchanged with a change in the position of the measured qubit, as indicated by 
Fig. \ref{fourq}(c), where results corresponding to measurement over qubit $2$ in the cases 
of $|\Psi^4_3\rangle$, $|\Psi^4_5\rangle$, and $|\Psi^4_6\rangle$ are presented. The LGGM of the states $|\Psi^4_i\rangle$, $i=7,8,9$, 
with single-qubit measurement, can be obtained analytically, and are tabulated in Table \ref{tab2}. 


\begin{figure}
 \includegraphics[scale=0.3]{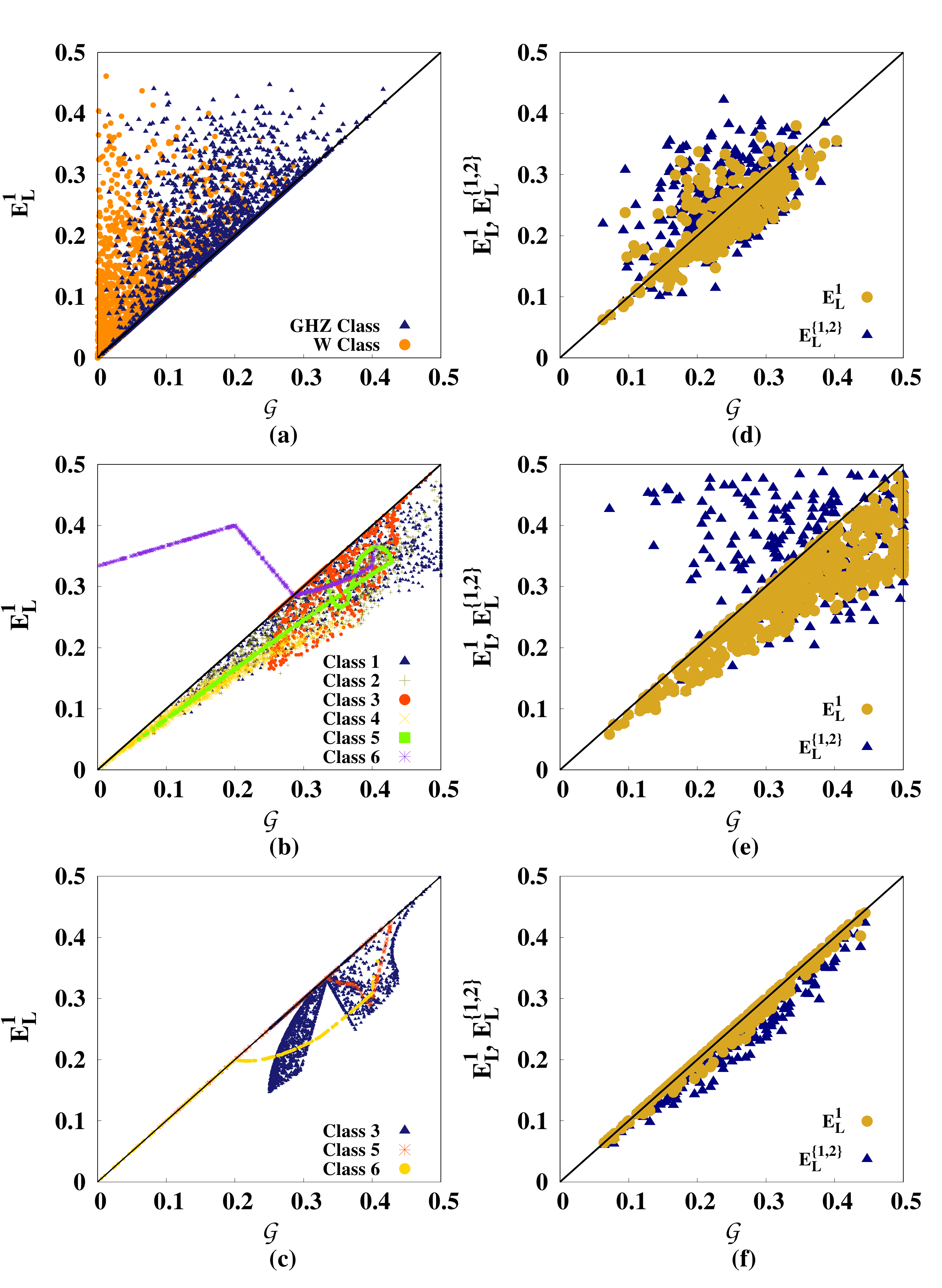}
 \caption{(Color online.) Scatter plots. 
 \textbf{(a)} $E_{L}^1$ vs. $\mathcal{G}$ for  three-qubit pure states belonging to the GHZ class and the W class. 
 \textbf{(b)} $E_L^1$ vs. $\mathcal{G}$ for the four-qubit classes represented by 
 $|\Psi^4_i\rangle$, $i=1,\cdots,6$. 
 \textbf{(c)} $E_L^2$ vs. $\mathcal{G}$ for the four-qubit classes represented by 
 $|\Psi^4_i\rangle$, $i=3,5,6$.
\textbf{(d)} $E_L^1$ and $E_L^{\{1,2\}}$ vs. $\mathcal{G}$ in the case of arbitrary four-qubit pure states. 
\textbf{(e)} $E_L^1$ and $E_L^{\{1,2\}}$ vs. $\mathcal{G}$ for four-qubit states of the form $|\Psi^4_1\rangle$.
and \textbf{(f)} $E_L^1$ and $E_L^{\{1,2\}}$ vs. $\mathcal{G}$ for four-qubit symmetric states of the form $|D^N_g\rangle$.
Each plot constitutes of $10^5$ Haar-uniformly generated pure states.
 All quantities plotted are dimensionless.}
 \label{fourq}
\end{figure}

\begin{figure}
 \includegraphics[scale=0.3]{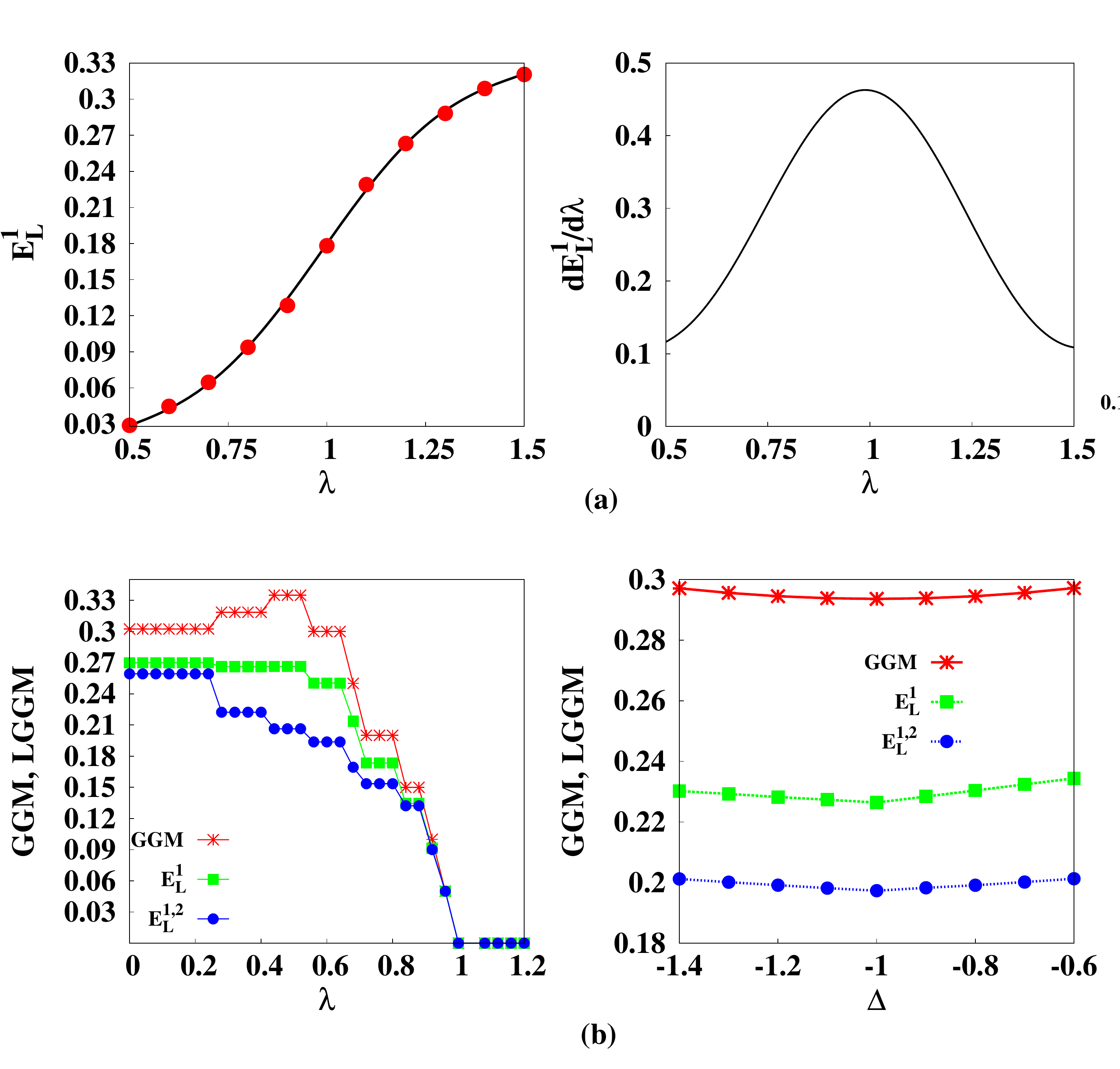}
 \caption{(Color online.) \textbf{(a)} Variations of $E_L^1$ (left panel), and $dE_L^{1}/d\lambda$ (right panel) as 
 functions of $\lambda=J/h$ for the transverse-field Ising model with $N=16$. \textbf{(b)} Variation of $\mathcal{G}$, $E_L^1$, and $E_L^{\{1,2\}}$ as functions of $\lambda$ with $\Delta=0.5$ (left panel) and $\Delta$ with $\lambda=0$ (right panel) for $N=20$. 
All quantities plotted are dimensionless.}
 \label{ising-fig-1}
\end{figure}



In the case of Haar-uniformly generated arbitrary four-qubit pure states, 
about $29\%$ of states are found to have $E_L^1>\mathcal{G}$, as depicted 
in Fig. \ref{fourq}(d). Interestingly, when local measurements over 
qubits $1$ and $2$ are performed, about $22\%$ of the $71\%$ states for which $\mathcal{G}>E^1_L$ are found to have 
$E^{\{1,2\}}_L>\mathcal{G}$. Also, for about $47.6\%$ of such states, $E^{\{1,2\}}_L>E_{L}^1$. 
Qualitatively similar results are found when random four-qubit states are sampled 
Haar-uniformly from the parametrized four-qubit classes \cite{slocc,infinite}. For example, 
the LGGM of about $20.4\%$ of the four-qubit states of the 
form $|\Psi^4_1\rangle$, for which $E_L^{1}\leq\mathcal{G}$, can be increased beyond 
the value of $\mathcal{G}$, when local measurements over 
qubits $1$ and $2$ are performed (Fig. \ref{fourq}(e)). 
Moreover, for about $44.9\%$ of the states, $E^{\{1,2\}}_L>E_{L}^1$. Hence, the results again indicate that local measurements on two parties can be more advantageous than the measurement on a single party for some multiparty states.  
The values of $E^{\{r_1,r_2\}}_L$, where $r_1\neq r_2$ with $r_{1,2}\in\{1,2,3,4\}$, 
in the case of the states $|\Psi^4_i\rangle$, 
$i=7,8,9$, are tabulated in Table \ref{tab2}. Note that the state $|\Psi^4_7\rangle$, 
provides an example of a four-qubit pure state for which $E_L^{\mathbf{r}_\alpha}=\mathcal{G}$ $\forall$ $\mathbf{r}_\alpha$, 
where $\alpha=1,2,\cdots,\binom{N}{m}$, $m=1,2$. %
Note also that $|\Psi^4_8\rangle$ presents an example of a four-qubit state for which a $E_L^r>\mathcal{G}$
only when measurement is performed over a specific qubit, and measurement over two qubits can enhance the value of 
LGGM over the single-qubit measurement.

It is evident from Figs. \ref{fourq}(b)-(e) that the arbitrary four-qubit pure states as well as four-qubit pure states of 
the form $|\Psi^4_i\rangle$, $i=1,\ldots,6$, 
are not uniformly distributed over the $\mathcal{G}-E_L^r$ and $\mathcal{G}-E_L^{\{r_1,r_2\}}$ planes, $r,r_1,r_2=1,\ldots,4$, 
but cluster around specific regions. 
In the case of $E_L^1$, one of the boundaries of the accessible region is always the line $E_L^1=\mathcal{G}$ when 
$|\Psi^4_i\rangle$, $i=1,\ldots,4$ are considered. For $|\Psi^4_5\rangle$ and $|\Psi^4_6\rangle$, the situation is different, and
all the states lie on very restricted regions on the $\mathcal{G}-E_L^r$ plane.
Such clustering is found in the case of five-qubit arbitrary pure states also, which highlights the importance of 
finding the quantum states that lies in a beneficial region, i.e., where $E_L^r \geq \mathcal{G}$, or $E_L^{\{r_1,r_2\}}\geq \mathcal{G}$. 
However, Haar-uniform numerical simulation of arbitrary pure states becomes difficult as the number of qubits increases, 
thereby restricting the investigation of 
the behavior of LGGM against the GGM of the original state in the case of $N>5$.


We conclude by mentioning that similar to the case of single-qubit measurement, in the case of $|D^N_g\rangle$ (Eq. (\ref{gen_sup_dicke})), 
$E^{\{r_1,r_2\}}_L\leq\mathcal{G}$, where $r_1\neq r_2$, and $r_{1,2}\in\{1,2,\cdots,N\}$. 
Note here that the symmetry present in the states $|D^N_g\rangle$ guarantees that all pairs of qubits chosen for local 
measurement are equivalent. Also, in contrast to the previous examples of two-qubit measurements, in this case,
an increase in $m$ is found to result, for an overwhelmingly large fraction of states, in a decrease in the value of LGGM. 
For example, in about $99.1\%$ of $|D^4_g\rangle$, $E^{\{r_1,r_2\}}_L<E^{1}_L$, while for $N=5$, the fraction is $93.9\%$. 
The variation of LGGM against GGM for $|D^4_g\rangle$ is depicted in Fig. \ref{fourq}(f). 

\section{LGGM in Quantum Spin Models}
\label{spin_lggm}

For the past fifteen years or so, probing interesting physical phenomena observed in many-body systems using 
quantum information theoretic tools and techniques has been an active cross-disciplinary 
field of research \cite{geom_ent_manybody,ggmn,fazio_rmp,modi_rmp,rev_xy,ent_xy_bp,ent_xy_mp,trap_rev}. 
Realization of quantum spin Hamiltonians in various substrates, such as solid state systems \cite{solid_xy}, 
optical lattice \cite{coldatom,xxz-exp,optical-lattice}, ion traps \cite{trap_rev,lab_xy_ion}, and 
NMR \cite{nmr} under controlled
laboratory environments have allowed researchers to test the properties of several information theoretic 
measures of quantum correlations in well-known quantum spin models. 
In this section, we discuss the behavior of LGGM in 
the vicinity of quantum phase transitions (QPT) \cite{qpt_book}, when the pure state, for which LGGM is calculated, is the ground state of 
well-known 1d quantum spin models. 

It is interesting to mention here that a similar set of results was obtained in Ref. \cite{loc_ent} by using the concept of localizable 
entanglement. An important difference of the results there with those in this paper is that the plots and the corresponding analyses here 
are directly for the LGGM, while those in Ref. \cite{loc_ent} are often for bounds on localizable entanglement. In this context, we would like 
to set emphasis on the fact that the computation of LGGM in a multiqubit state with high number of qubits can also be involved. On one hand, one has to perform an optimization over $2m$ real parameters (see discussion in Sec. \ref{sqm}) with $m$ being the number of measured qubits. 
On the other hand, in order to calculate LGGM corresponding to local measurement over $m$ qubits in an $N$-qubit pure state, 
GGM of $(N-m)$-partite pure states in the post-measurement ensemble needs to be computed. Typically, for an $\tilde{N}$-qubit pure state, 
computation of GGM in its full generality requires consideration of the maximum eigenvalues of 
a total of $\sum_{i=1}^{\tilde{N}/2}\binom{\tilde{N}}{i}$ density matrices of dimensions $2^i$, $i=1,\ldots,\tilde{N}/2$, indicating an exponential increase of 
the computational complexity of LGGM with increasing $\tilde{N}$. Also, exact computation of LGGM for the ground state of a spin Hamiltonian 
constituted of $N$ spins requires access to the exact ground state of the Hamiltonian, the determination of which is non-trivial 
for high values of $N$. Therefore, we use Lanczos diagonalization technique \cite{titpack} to determine the ground states of the spin Hamiltonians when 
$N$ is large, and 
use the symmetries of the ground state to reduce the computational complexity. 

\subsection{Transverse-field Ising Model}

The first model we discuss is the transverse-field quantum Ising model \cite{xy_group} in 1d, with periodic 
boundary condition (PBC), whose Hamiltonian is given by 
\begin{eqnarray}
H&=& J\sum_{i=1}^{N}\sigma_{i}^{x}\sigma_{i+1}^{x}
+h\sum_{i=1}^{N}\sigma_{i}^{z}.
\label{ising}
\end{eqnarray}
Here, $J$ and $h$ respectively  are strengths of the nearest-neighbor exchange coupling and the external magnetic field,
$\sigma_i^{x,z}$ are the Pauli spin matrices, and $N$ is the number of qubits in the system. Under PBC, 
$\sigma_{N+1}\equiv\sigma_1$. Both $J$ and $h$ are chosen to be positive. 
The model is known to undergo a QPT 
at the critical value of the parameter $\lambda=\lambda_c\equiv 1$ 
\cite{fazio_rmp,modi_rmp,rev_xy,xy_group,qpt_book,ent_xy_bp}, where $\lambda=J/h$,  from an anti-ferromagnetic ($J>h$) 
to a paramagnetic phase ($J<h$). A few studies of the behavior of multipartite 
entanglement measures across QPT in quantum spin models are available \cite{ent_xy_mp,geom_ent_manybody,ggmn} due to the difficulty of computing such measures. However, the investigation of bipartite measures has extensively been carried out \cite{trap_rev,fazio_rmp}. 

To overcome the difficulty in computing LGGM for higher number of parties, 
we first look into systems comprised of a relatively smaller number of parties (say as $N=8,10,12$), and compute the LGGM of the ground state without 
any approximation. We find that irrespective of the value of $N$, in the case of single-qubit measurement, the maximization 
involved in Eq. (\ref{lmen}) can be achieved by using the computational basis $\{|0\rangle,|1\rangle\}$ (see also \cite{regular_point}).
Moreover, in the case of single-qubit measurement, the GGM of the $(N-1)$-qubit pure states in the post-measurement ensemble is always 
obtained from either the $1$:rest, or the $2$:rest bipartitions. We use these information to  
compute the LGGM for the ground state of the system with high values of $N$. Note that the 1d transverse-field Ising model can be solved by 
successive application of the Jordan-Wigner and Bogoliubov transformations \cite{xy_group}. However, this does not provide access to an 
analytical form of the ground state, from which LGGM can be computed. We determine the ground state of the model with high values of $N$ 
via the Lanczos diagonalization scheme, and investigate the behavior of LGGM close to the QPT point. 
%
%
%
Fig. \ref{ising-fig-1}(a) (left panel) depicts the variation of $E_L^1$ against $\lambda$ in the ground state of the model with
$N=16$. The points
in the graph represents the numerical values of $E^1_L$ obtained by using the approximations discussed above, while the continuous
line represents the fitted curve. The first derivative of $E_L^1$ with $\lambda$, as obtained from the fitted curve, shows a maximum at the 
QPT point $\lambda_c=1$ (Fig. \ref{ising-fig-1}(a), right panel). The maximum sharpens with increasing $N$.

\subsection{XXZ model in an external field}

The Hamiltonian describing the 1d XXZ model \cite{bethe,xxz_group,mikeska,takahashi}, under PBC, is given by 
\begin{eqnarray}
H^\prime&=& J^\prime\sum_{i=1}^{N}(\sigma_{i}^{x}\sigma_{i+1}^{x}+\sigma_i^y\sigma_{i+1}^y-\Delta\sigma_i^z\sigma_{i+1}^z)
+h^\prime\sum_{i=1}^{N}\sigma_{i}^{z},\nonumber \\
\label{xxz}
\end{eqnarray}
where $J^\prime$ and $h^\prime$  represent respectively the nearest-neighbor exchange coupling and the external field-strength, 
while $\Delta$ is the (dimensionless) anisotropy 
in the $z$ direction. Although the model can be solved using the thermodynamic Bethe ansatz technique \cite{bethe,takahashi}, 
an analytical form of the LGGM or any multipartite entanglement measure is still elusive due to the inaccessibility of the analytical form of the  exact ground state.
Similar to the case of the transverse Ising model, we probe the phase diagram \cite{takahashi} of this model using LGGM 
as the physical quantity, where the ground state is obtained 
via Lanczos diagonalization method for different values of the system parameters, $\Delta$ and $h^\prime/J^\prime$. 
The observations in the transverse-field Ising model, such as obtaining the GGM of the post-measurement pure states from either the $1$:rest, or the $2$:rest bipartitions, in the case of smaller number of spins, are found to be valid here even for local measurement over two qubits. An assumption of the validity of these observations in the case of large 
$N$, considerably reduces the complexity of computation.  
Fig. \ref{ising-fig-1}(b) (left panel) shows the variations of $\mathcal{G}$, $E_{L}^{1}$, and 
$E_{L}^{\{1,2\}}$ against $\lambda\equiv h^\prime/J^\prime$ for the ground state of the 
Hamiltonian given in Eq. (\ref{xxz}), with $N=20$ and $\Delta=0.5$. 
With increasing $\lambda$, all the quantum correlation measures become zero at $\lambda=\lambda_c\equiv 1$, signaling a transition of the ground state from an entangled to a product state \cite{xxz_group,takahashi}. The $z$-component of the total spin is a conserved quantity of the system due to the presence of ${\mathbb{Z}}_2$ symmetry in the Hamiltonian, thereby making the Hamiltonian block-diagonalizable. The plateaus in the variations of GGM and LGGM with $\lambda$ correspond to different values of the  $z$-component of the total spin.
With increasing $N$, the number of the plateaus increases, while the widths of the individual plateaus decrease, and the curves eventually tend to 
continuous ones for high values of $N$.

In the absence of the external magnetic field, the ground state of the model experiences a 
KT transition \cite{xxz_group} at $\Delta=-1$, which is signaled by a blunt minimum in the $\mathcal{G}$ vs. $\Delta$ curve 
(see Fig. \ref{ising-fig-1}(b) (right panel)). On the other hand, the LGGMs, namely $E_{L}^1$ and $E_{L}^{\{1,2\}}$, show a sharper cusp at $\Delta=-1$, 
correctly signaling the KT transition. 
Therefore, similar to 
localizable entanglement \cite{loc_ent}, LGGM also signals 
the KT transition, which is usually not detected by the frequently-used 
quantum information theoretic quantities in detecting 
QPTs \cite{oth_mes_xxz}.

Note here that even for a small system where finite-size effects are expected to play a considerable role, 
LGGM serves as a satisfactory indicator of quantum critical phenomena. 
Hence, LGGM is expected to find its applicability in the investigation of 
quantum cooperative phenomena observed in many-body systems beyond the quantum spin models.

\section{Concluding remarks}
\label{conclude}

Multipartite entanglement has been proven to be useful for the successful implementation of several quantum 
information theoretic protocols. In the present paper, we consider the conceptualization and use of localizable multipartite 
entanglement, obtained by performing a local measurement at a few parties. The notion of localizable multiparty entanglement depends on 
the understanding and computability of another measure of multiparty entanglement of a lower number of parties, which we refer to as the 
seed measure. We use the geometric measure in general, and the generalized geometric measure in particular, as the see measure. In the case when 
the seed measure is the generalized geometric measure, the localizable multiparty entanglement measure is called the localizable generalized 
geometric measure. We discuss its various properties, and in particular, we analytically consider the 
behavior of the measure for a number of paradigmatic examples of 
multipartite pure states, where the localization is achieved via local rank-$1$ 
projective measurement over a single qubit. 
The examples include the $N$-qubit generalized GHZ and W states, Dicke states, and the generalized 
superposition of Dicke states
for fixed number of qubits. We show that for the $N$-qubit generalized GHZ and W states, 
the localization of generalized geometric 
measure by local projection measurement over one qubit always results 
in a value of localizable multipartite entanglement which is 
greater than, or equal to the multipartite entanglement present in the original 
state. Our numerical simulations seem to indicate that such enhancement due to measurement holds
for arbitrary three-qubit pure states. However, for higher number of parties, 
no such lower bound exists. We also show that for the 
$N$-qubit Dicke states, the localizable multipartite entanglement achieved via 
single-qubit measurement is bounded above by the 
generalized geometric measure of the original state. 

To investigate whether measurement over more than one qubit helps in achieving a 
better localization of multipartite entanglement
than a single-qubit measurement, we consider several examples of multi-qubit pure states. 
We show that there exists multiqubit states
in which local measurement over two qubits yields higher values of 
localizable generalized geometric measure compared to 
single-qubit measurement. 
However, this phenomenon is not generic for arbitrary 
multiqubit states, as shown from our numerical simulations.
We finally inquire whether there exists a situation in which LGGM is more powerful than its parent 
multipartite entanglement measure, the GGM. We show that this is indeed the case for detecting 
QPT in the many-body systems. Specifically, we show that the derivatives of LGGM can signal QPT of the transverse-field 
Ising model more accurately even for a smaller system-size, achievable in current experiments, compared to that of the GGM.
We also show that LGGM detects the KT transition of the XXZ model better than its parent multiparty measure.

\appendix
 
\section{Multipartite entanglement measures}
\label{mult-mes}
 
The geometric measure of entanglement of an $N$-partite state $|\tilde{\Psi}_N\rangle$, 
is defined as   
\begin{eqnarray}
 G_{K}(|\tilde{\Psi}_N\rangle)=1-\underset{\mathcal{S}_K}{\max}|\langle\tilde{\Phi}_N^{K}|\tilde{\Psi}_N\rangle|^2,
 \label{gm_ksep}
\end{eqnarray}
where $K$ ($2\leq K\leq N$) is an integer denoting the number of product state partitions into which the $N$-partite state 
$|\tilde{\Phi}_N^{K}\rangle$ can be 
divided \cite{horodecki_rmp,geom_ent,geom_ent_manybody,geom_ent_ksep,geom_ent_book,ggm_first,ggmn}. 
The distance of the state $|\tilde{\Psi}_N\rangle$ is minimized over the set, 
$\mathcal{S}_K$, of all $K$-separable pure states, $|\tilde{\Phi}_N^{K}\rangle$, and we refer to this measure as K-GM.
For example, $K=N$ corresponds to a fully separable state 
$|\tilde{\Phi}_N^{N}\rangle\equiv\bigotimes_{i=1}^{N}|\tilde{\phi}_{i}\rangle$, leading to the original definition of 
GM \cite{geom_ent,geom_ent_book}.
On the other hand, GGM is obtained for $K=2$ \cite{ggm_first,ggmn}, 
the other extremum of K-GM, which we denote by $\mathcal{G}(|\tilde{\Psi}_N\rangle)$. 
The optimization in the definition of $\mathcal{G}(|\tilde{\Psi}_N\rangle)$ can be performed  by  
using the maximization of the Schmidt coefficients across all possible bipartitions of $|\tilde{\Psi}_N\rangle$, leading 
to the simplified version of GGM \cite{ggmn}, given by
\begin{eqnarray}
 \mathcal{G}(|\tilde{\psi}_N\rangle)
 =1-\underset{\mathcal{S}_{\mathcal{A}:\mathcal{B}}}{\max}\{\lambda^2_{\mathcal{A}:\mathcal{B}}\}.
 \label{ggm_exp}
\end{eqnarray}
Here, $\lambda_{\mathcal{A}:\mathcal{B}}$ is the maximum Schmidt 
coefficient of $|\tilde{\Psi}_N\rangle$, the maximum being 
taken over the set, $\mathcal{S}_{\mathcal{A}:\mathcal{B}}$, of all arbitrary $\mathcal{A}:\mathcal{B}$
bipartitions such that $\mathcal{A}\cup\mathcal{B}=\{1,2,\cdots,N\}$, 
and $\mathcal{A}\cap\mathcal{B}=\mathbf{\Phi}$, the null set. 
The above expression allows one to compute GGM of a multiparty 
pure state in arbitrary dimensions and for arbitrary number of parties.

\section{Proofs of the Propositions}
\label{prop_proof} 

\noindent\textbf{Proposition \ref{pr:ghz}.} Once the rank-$1$ projective measurement is 
performed over the first qubit, an ensemble of two pure states, 
$\{p^{l},|\Psi^l\rangle\}$, $l=1,2$, 
is obtained, where $p^{l}=|a_1|^2q^l_1+|a_2|^2q^l_2$, 
$|\Psi^{l}\rangle=|\xi^{l}_1\rangle\otimes|\psi^{l}\rangle$ for $l=1,2$,  and
$|\psi^{l}\rangle=(a_1z^l_1|0\rangle^{\otimes N-1}+a_2z^l_2|1\rangle^{\otimes N-1})/\sqrt{p^l}$. The quantities 
$q^l_{1,2}$ and $z^l_{1,2}$ are given by
\begin{eqnarray}
 q^{l}_1&=&(\delta_{1l}c_{\theta/2}^{2}+\delta_{2l}s_{\theta/2}^{2}),\nonumber \\
 q^{l}_2&=&(\delta_{1l}s_{\theta/2}^{2}+\delta_{2l}c_{\theta/2}^{2}),
 \label{ql12}
\end{eqnarray}
and
\begin{eqnarray} 
z^l_1&=&\delta_{1l}c_{\theta/2}-\delta_{2l}e^{-i\phi}s_{\theta/2}, \nonumber \\ 
z^l_2&=&\delta_{1l}e^{i\phi}s_{\theta/2}+\delta_{2l}c_{\theta/2},
\label{zl12}
\end{eqnarray}
with $\delta_{kl}$, $k=1,2$, being Kronecker delta.
Note that the states $|\psi^l\rangle$ are of the form of an $(N-1)$-qubit gGHZ state.
Note also that for an $N$-qubit gGHZ state, the local density matrix, $\rho_n$, of $n$ qubits, 
is diagonal in the computational basis with only two non-zero elements.  Hence, in the present case, 
the LGGM of the gGHZ state can be obtained from Eq. (\ref{lmen}), where 
$\mathcal{G}(\psi^l_{N-1})=1-\max\left\{|a_1|^2q^l_1/p^l,|a_2|^2q^l_2/p^l\right\}$ with $p^{l}=|a_1|^2q^l_1+|a_2|^2q^l_2$.
Using the identity, $\max\{x,y\}+\max\{u,v\}=\max\{x+u,x+v,y+u,y+v\}$ for arbitrary values of $x$, $y$, $u$, and $v$, 
in Eq. (\ref{lmen}), one obtains 
\begin{eqnarray}
 E_{L}^1=1-\underset{\theta}{\min}\left[\max\left\{|a_1|^2,s_{\theta/2}^{2},c_{\theta/2}^{2}\right\}\right].
\label{lggm_gghz}
\end{eqnarray}

In $[0,\pi]$, $c_{\theta/2}^{2}$ is a monotonically decreasing 
function whereas $s_{\theta/2}^{2}$ is a monotonically increasing one, 
and the maximum value of both the functions is unity, occurring at $\theta=0$ 
for $c_{\theta/2}^{2}$, and $\theta=\pi$ 
for $s_{\theta/2}^{2}$. Besides, the values of both functions are equal to $1/2$ at $\theta=\pi/2$. 
Since $|a_1|^2\geq 1/2 \geq|a_2|^2$,
the allowed range of $\theta$ can be divided into three subregions: 
\textbf{(i)} $0\leq\theta\leq\theta_{1}$, where $c_{\theta/2}^{2}\geq|a_1|^2\geq s_{\theta/2}^{2}$, \textbf{(ii)} 
$\theta_{1}<\theta<\theta_2$, where $|a_1|^2\geq\max\{c_{\theta/2}^{2},s_{\theta/2}^{2}\}$, and \textbf{(iii)} 
$\theta_2\leq\theta\leq\pi$, where $s_{\theta/2}^{2}\geq|a_1|^2\geq c_{\theta/2}^{2}$.
Here, the values of $\theta_1$ and $\theta_2$ are obtained as solutions of the equations 
$c^2_{\theta/2}=|a_1|^2$ and $s^2_{\theta/2}=|a_1|^2$ respectively. 
The maximization inside the curly bracket in Eq. (\ref{lggm_gghz})
has to be performed for a fixed value of $\theta$, which  can be chosen from any one of the three subregions, 
\textbf{(i)}, \textbf{(ii)}, and \textbf{(iii)}. Noticing that 
$\min(c^2_{\theta/2})=|a_1|^2$ for $0\leq\theta\leq\theta_1$,
while $\min(s^2_{\theta/2})=|a_1|^2$ in the region $\theta_2\leq\theta\leq\pi$, 
from Eq. (\ref{lggm_gghz}),  one obtains  
$E_{L}^1=|a_2|^2$. Note here that the GGM of the gGHZ state, under the assumption 
that $|a_{1}|^2\geq\frac{1}{2}\geq|a_2|^2$, is 
also $|a_2|^2$, leading to $E_L^{1}=\mathcal{G}$. \hfill $\blacksquare$

\noindent\textbf{Proposition \ref{pr:3gw}.} The post measurement ensemble is 
$\{p^l=(|a_2|^2+|a_3|^2)q^l_1+|a_1|^2q^l_2,
|\psi^l\rangle=[z^l_1(a_2|10\rangle+a_3|01\rangle)+z^l_2a_1|00\rangle]/\sqrt{p^l}\}$, $l=1,2$, 
where $z^l_{1,2}$ 
and $q^l_{1,2}$ are given in Eqs. (\ref{ql12}) and (\ref{zl12}) respectively. One must note that the GGM of two-qubit 
states of the form $|\psi^l\rangle$
is invariant under an interchange of $a_2$ and $a_3$. The LGGM (Eq. (\ref{lmen})), in this case, is given by 
\begin{eqnarray}
 E_L^1=\frac{1}{2}\left(1-\underset{\theta}{\min}\sum_{l=1}^{2}\sqrt{f^l(\theta)}\right),
 \label{fgw}
\end{eqnarray}
where$f^l(\theta)={p^l}^2-4|a_2|^2|a_3|^2u^l$, and 
\begin{eqnarray}
u^l=\delta_{1l}c_{\theta/2}^4+\delta_{2l}s_{\theta/2}^4.
\label{ul} 
\end{eqnarray}
The validity of the single qubit density matrix $\rho^l_2$ of qubit $2$, obtained from the state $|\psi^l\rangle$,
demands that $f^l(\theta)\geq0$ $\forall$ $\theta$. The function $f(\theta)=\sum_{l=1}^{2}\sqrt{f^l(\theta)}$ has minimum at $\theta=0,\pi$
in the interval $0<\theta<\pi$, while having a maximum at $\theta=\pi/2$.  
Hence the optimization in LGGM, in this case, is achieved when measurement is performed in 
the $\{|0\rangle,|1\rangle\}$ basis. 
From Eq. (\ref{fgw}), it can be shown that $E_L^1=\min\{|a_2|^2,|a_3|^2\}$. 
In general, for the local projective measurement being performed on the qubit $r$, 
$E_L^{r}=\min\{|a_j|^2,|a_k|^2\}$, where $r\neq j\neq k$, and $r,j,k\in\{1,2,3\}$.  
This leads us to a non-trivial lower bound of LGGM for the tripartite gW state, 
which is invariant to a change in the choice of 
the measured qubit, as given by the following proposition.

We first consider $r=1$. Let us first assume $|a_1|^2\geq|a_2|^2\geq|a_3|^2$, which 
leads to $\mathcal{G}=\min\{|a_j|^2\}$, $j=2,3$. Assuming other orderings and considering all the cases, we obtain
$\mathcal{G}=\min\{|a_i|^2\}$, $i=1,2,3$.
If $|a_i|^2\in\{|a_2|^2,|a_3|^2\}$, then $E^{1}_L=\mathcal{G}$. 
Else, $\mathcal{G}=|a_1|^2$, implying $E_L^{1}=\min\{|a_2|^2,|a_3|^2\}\geq\mathcal{G}$. Similar proofs hold when $r=2,3$. \hfill $\blacksquare$

\noindent\textbf{Proposition \ref{pr:gw_ubound}.} As in the case of three-qubit systems, here also we start from the case $r=1$. 
The GGM of the state
$|W_N\rangle_g$ can be obtained as $\mathcal{G}=1-\max\{\lambda^{max}_{n}\}$, where 
$\lambda^{max}_{n}$ is the maximum eigenvalue of all possible $n$-qubit reduced density matrices, $\rho^{\mathbf{s}}_{n}$, 
where $1\leq n\leq N/2$ ($1\leq n\leq (N-1)/2$) for even (odd) $N$. Here, ``$\mathbf{s}$'' denotes the set of all possible 
indices $\{s_j\}$, $j=1,\cdots,n$, that represents the positions of the $n$ qubits with $s_j\in\{1,2,\cdots,N\}$. 
The density matrices, $\rho_n^{\mathbf{s}}$, can be written as
\begin{eqnarray}
 \rho_n^\mathbf{s}&=&
 P\left[\sum_{j=1}^na_{s_j}|0\rangle^{\otimes(j-1)}|1\rangle_j|0\rangle^{\otimes(N-j)}\right]\nonumber \\
&&+\sum_{\underset{k\notin\mathbf{s}}{k=1}}^N|a_k|^2P[|0\rangle^{\otimes n}],
\label{redmat_gw_n}
\end{eqnarray}
where $P[|a\rangle]=|a\rangle\langle a|$. From Eq. (\ref{redmat_gw_n}), it is clear that the maximal eigenvalue 
is obtained from the case $n=1$, and from the single qubit density matrix for which 
$|a_{s_1}|^2=\min\{|a_i|^2\}$, $i=1,2,\cdots,N$, leading to $\mathcal{G}=|a_{s_1}|^2$. 

To determine the LGGM, we first consider a measurement in the 
computational basis, $\{|0\rangle,|1\rangle\}$. 
A measurement over the first qubit in this basis 
leads to a product state, $|0\rangle^{\otimes (N-1)}$, with 
probability $p_0=|a_1|^2$, and a pure state $|\Phi^{N-1}\rangle$, with probability $\sum_{i=2}^N|a_i|^2$, which can be 
identified as an $(N-1)$-qubit gW state. From the
above discussion, $\mathcal{G}(|\Phi^{N-1}\rangle)=|a_{j}|^2/\sum_{i=2}^{N}|a_i|^2$, where 
$|a_j|^2=\min\{|a_{i}|^2\}$, $i=2,3,\cdots,N$. The definition of LGGM implies $E_{L,0}^1=|a_j|^2$, 
where the subscript ``0'' 
indicates that the measurement is performed in the basis $\{|0\rangle,|1\rangle\}$. Clearly, 
$|a_j|^2\geq|a_{s_1}|^2\equiv\mathcal{G}$. 
Since the definition of LGGM involves a maximization over the complete set of
projective measurements, $E_{L}^1\geq E_{L,0}^1\geq\mathcal{G}$. Similarly, one can prove for the cases $r=2,3,\cdots,N$,
and hence the proof. \hfill $\blacksquare$

\noindent\textbf{Proposition \ref{pr:spgw}.} Let us consider a gW state with the ordering 
$|a_1|^2\leq|a_2|^2\leq\cdots\leq|a_N|^2$. Assuming local projective measurement in qubit $1$,
$\mathcal{G}=|a_1|^2$ and $E_L^1=|a_2|^2$. The assumed ordering suggests that 
$\max\{|a_2|^2\}=(1-|a_1|^2)/(N-1)$, which corresponds to the LGGM of a gW state of the form 
$|\Phi\rangle=|a_1||1\rangle|0\rangle^{\otimes N-1}+
\sum_{j=2}^{N}a_j|0\rangle^{\otimes(j-1)}|1\rangle_j|0\rangle^{\otimes(N-j)}$, with $\mathcal{G}=|a_1|^2$,
similar to the arbitrary gW state, 
and $|a_j|^2=(1-|a_1|^2)/(N-1)$, $j=2,3,\cdots,N$. Since $\mathcal{G}=|a_i|^2=\min\{|a_k|^2\}$ 
for an arbitrary $N$-qubit gW state with arbitrary ordering of $\{|a_k|^2\}$, $k=1,2,\cdots,N$, one can prove similar 
result for each possible ordering of $\{|a_k|^2\}$, when measurement over qubit $1$ is assumed. The result also 
holds for an arbitrary position, $r$, of the measured qubit, when $r\in\{1,2,\cdots,N\}$. Hence the proof. 
\hfill $\blacksquare$

\section{Single-qubit measurement on symmetric states}
\label{ap:dicke}

Since Dicke states are symmetric, the value of 
LGGM is independent of the position of the 
qubit over which measurement is performed, and the  post measurement ensemble, $\{p^l,\psi^l\}$, is given by
$p^l=A_kq^l_1+B_kq^l_2$ and $|\psi^l\rangle=\sqrt{A_k}z^l_1|D^{N-1}_k\rangle+\sqrt{B_k}z^l_2|D^{N-1}_{k-1}\rangle$, 
with $A_k=\binom{N-1}{k}/\binom{N}{k}$, and $B_k=\binom{N-1}{k-1}/\binom{N}{k}$, 
while $q^l_{1,2}$ and $z^l_{1,2}$ being defined in Eqs. (\ref{ql12}) and (\ref{zl12}).
To determine the LGGM of $|D^N_k\rangle$, one needs 
to calculate the GGM of $|\psi^l\rangle$, which, in turn, requires determination of the 
reduced density matrix of $n$ qubits, $\rho_n^{l,\mathbf{s}}$,  labeled with the set of indices 
$\mathbf{s}\equiv\{s_1,s_2,\cdots,s_n\}$. Since measurement over any 
single qubit of $|D^N_k\rangle$ yields an ensemble of symmetric 
states, the reduced density matrix, $\rho_n^{l,\mathbf{s}}$, of all possible 
collection of $n$ qubits of the state $|\psi^l\rangle$, 
are equivalent.  
Therefore, discarding the index ``$\mathbf{s}$'', $\rho_n^{l}$ can be obtained from $|\psi^l\rangle$ as 
\begin{eqnarray}
 \rho_n^l&=&\frac{1}{\binom{N}{k}p^l}\Big[\sum_{i=0}^nF_i^l P[|D^n_i\rangle]
 -(-1)^l\sum_{i=0}^{n-1}G_i\Big(e^{-i\phi}|D^n_{i+1}\rangle\langle D^n_i| \nonumber \\
 &&+e^{i\phi}|D^n_i\rangle\langle D^n_{i+1}|\Big)\Big],
\end{eqnarray}
with $1\leq n\leq N^\prime$, where $N^\prime=(N-2)/2$ $(N^\prime=(N-1)/2)$ when $N$ is even (odd), and
\begin{eqnarray}
 F^l_i&=&\binom{n}{i}\left[\binom{N-n-i}{k-i}q^l_1+\binom{N-n-i}{k-i-1}q^l_2\right]\nonumber \\
 G_i&=&\frac{s_\theta}{2}\binom{N-n-1}{k-i-1}\sqrt{\binom{n}{i+1}\binom{n}{i}}.
\end{eqnarray}
The GGM of $|\psi^l\rangle$ is given by 
$\mathcal{G}^l=1-\max\{\Lambda^l_n\}$, $n=1,2,\cdots,N^\prime$, where $\Lambda^l_n$ is the maximum 
eigenvalue of $\rho^l_n$. 

Similarly, for an arbitrary state $|D^N_g\rangle$ of the form (\ref{gen_sup_dicke}), rank-$1$ projective 
measurement over qubit $r$, $r=1,2,\cdots,N$, produces an ensemble of two $(N-1)$-qubit symmetric states represented by 
$\{p^l,|\bar{D}^{N-1}_g\rangle\}$, where $p^l=\sum_{k=1}^{N-1} \binom{N-1}{k}|\bar{z}^l_1a_k+\bar{z}^l_2a_{k+1}|^2$,
$|\bar{D}^{N-1}_g\rangle=\frac{1}{\sqrt{p^l}}\sum_{k=1}^{N-1}(\bar{z}^l_1a_k+\bar{z}^l_2a_{k+1})|D^{N-1}_k\rangle$, with
$\bar{z}^l_1=\delta_{1l}c_{\theta/2}-(-1)^l\delta_{2l}e^{i\phi}s_{\theta/2}$, and  
$\bar{z}^l_2=\delta_{1l}e^{-i\phi}s_{\theta/2}+\delta_{2l}c_{\theta/2}$.

\section{LGGM with local measurement over more than one qubit}
\label{ap:example}

Here we present two more examples of multipartite pure states for which local measurement over more than one qubit may turn out to be 
beneficial regarding the value of LGGM. 

\noindent\textbf{Example 1.} Consider the four-qubit state given by
$|\Psi_4\rangle=a(|0000\rangle+|0011\rangle+|1100\rangle+|1111\rangle)
+\sqrt{(1-4a^2)/6}(|0101\rangle+|1010\rangle+|0110\rangle+|1001\rangle+|1011\rangle+|0100\rangle)$,
where $a\leq 1/2$, $a$ being a real number. Note that unlike the four-qubit state considered in Sec. \ref{mqm}, this state
does not belong to the set of gW states. However, similar to the former case, here also, $E_{L}^{\{1,2\}}>E_{L}^1$ for a finite 
range of the allowed values of the state-parameter $a$. 

\noindent\textbf{Example 2.} Consider the five-qubit state given by
$|\Psi_5\rangle=a(|00000\rangle+|00111\rangle+|11000\rangle+|11111\rangle)
+\sqrt{(1-4a^2)/4}(|01010\rangle+|10101\rangle+|00001\rangle+|10000\rangle)$.
Here again a finite parameter range can be obtained in which $E_{L}^{\{1,2\}}>\mathcal{G}$, although $E_L^1<\mathcal{G}$. 

These examples highlight the importance of two-qubit 
measurement in the cases where single-qubit measurement is not enough to increase the 
multipartite entanglement possessed by the original state. 
The results of our investigation on the existence of arbitrary 
multiqubit pure states, in which two-qubit measurements may yield better results than single-qubit ones, 
are presented in Sec. \ref{num_res}.
Intuitively, one can argue that if one increases the number of parties in which the measurement is performed, it helps to 
concentrate entanglement and hence to increase LGGM. Although the above examples support such intuition, counter-examples also exist.

%

\section{GHZ and W class of states}
\label{ap:ghzw}
The normalized three-qubit states of the W-class, up to LU, are given by \cite{dvc}
\begin{equation}
 |\Phi_{w}\rangle=\sqrt{a_1}|001\rangle+\sqrt{a_2}|010\rangle+\sqrt{a_3}|100\rangle+\sqrt{a_4}|000\rangle
 \label{state_wclass}
\end{equation}
with $a_1,a_2,a_3>0$, and $a_4=1-(a_1+a_2+a_3)\geq0$.
Simple algebra dictates that the GGM of $|\Phi_w\rangle$ is given by 
$\mathcal{G}=1-\max\{\lambda_i\}$, $i=1,2,3$, where $\lambda_i=[1+(1-4((a_j+a_k)a_i)^{1/2}]/2$,
with $i,j,k\in\{1,2,3\}$, and no two among $i,j,k$ being equal.
Now, LGGM of $|\Phi_w\rangle$  
for $r=1$ can be obtained as 
$E_{L}^1=[1-\underset{\theta,\phi}{\min}f_{wc}(\theta,\phi)]/2$,
where  
\begin{eqnarray}   
 f_{wc}(\theta,\phi)=\sum_{l=1}^{2}({p^l}^2-4a_1a_2u^l)^{1/2},
\end{eqnarray}
and $p^l=(a_1+a_2+a_4)q^l_1+a_3q^l_2-(-1)^l\sqrt{a_3a_4}s_{\theta}c_{\phi}$, with $q^l_{1,2}$ and $u^l$ given in 
Eqs. (\ref{ql12}) and (\ref{ul}), respectively. The optimization of the above function leads to 
two equations involving the real parameters $\theta$ and $\phi$. One of them implies $\phi=0,\pi$ independent of the value of $\theta$, 
while the second equation, independent of $\phi$, 
has to be solved numerically for $\theta$. Numerical solution of the latter provides the values of $\theta$, which, 
along with $\phi=0,\pi$, makes the Jacobian of $f_{wc}(\theta,\phi)$ positive semidefinite. 

The normalized three-qubit states of the GHZ class, up to LU, can be represented by \cite{dvc}
\begin{equation}
 |\Phi_{ghz}\rangle=\sqrt{K}\left(c_\delta|000\rangle+e^{i\mu}s_\delta\bigotimes_{i=1}^{3}|\eta_i\rangle\right), \nonumber
\end{equation}
where $|\eta_i\rangle=c_{\gamma_i}|0\rangle+s_{\gamma_i}|1\rangle$, 
and $K^{-1}=1+2c_{\delta}s_{\delta}c_{\gamma_1}c_{\gamma_2}c_{\gamma_3}c_{\mu}$, 
$K$ being the normalization factor,  and 
$K \in (1/2,\infty)$. The ranges for the five real parameters are $\delta\in(0,\pi/4]$, $\gamma_i\in(0,\pi/2]$, $i=1,2,3$, 
and $\mu\in[0,2\pi)$. Due to increased number of parameters, determination of GGM as well as LGGM for arbitrary GHZ class of states
are to be achieved via numerical techniques.

\section{Classes of four-qubit states}
\label{ap:fourq}
The nine classes of four-qubit states, as considered in \cite{slocc,infinite}, are 
\small
\begin{widetext}
\begin{eqnarray}
 |\Psi^4_1\rangle&=&\frac{1}{2}\{(a_1+a_2)(|0000\rangle+|1111\rangle)+(a_1-a_2)(|0011\rangle+|1100\rangle)
 +(a_3+a_4)(|0101\rangle+|1010\rangle)+(a_3-a_4)(|0110\rangle+|1001\rangle)\},\nonumber \\
 |\Psi^4_2\rangle&=&\frac{1}{2}\{(a_1+a_2)(|0000\rangle+|1111\rangle)+(a_1-a_2)(|0011\rangle+|1100\rangle)
 +2a_3(|0101\rangle+|1010\rangle+|0110\rangle),\nonumber \\ 
 |\Psi^4_3\rangle&=&a_1(|0000\rangle+|1111\rangle)+a_2(|0101\rangle+|1010\rangle+|0110\rangle+|0011\rangle), \nonumber \\
 |\Psi^4_4\rangle&=&\frac{1}{2}\{2a_1(|0000\rangle+|1111\rangle)+(a_1+a_2)(|0101\rangle+|1010\rangle)
 +(a_1-a_2)(|0110\rangle+|1001\rangle)+\sqrt{2}i(|0001\rangle+|0010\rangle+|0111\rangle+|1011\rangle)\},\nonumber \\ 
 |\Psi^4_5\rangle&=&a_1(|0000\rangle+|0101\rangle+|1010\rangle+|1111\rangle)
 +i|0001\rangle+|0110\rangle-i|1011\rangle,\nonumber \\ 
 |\Psi^4_6\rangle&=&a_1(|0000\rangle+|1111\rangle)+|0011\rangle+|0101\rangle+|0110\rangle,\nonumber \\
 |\Psi^4_7\rangle&=&|0000\rangle+|0101\rangle+|1000\rangle+|1110\rangle,\nonumber \\ 
 |\Psi^4_8\rangle&=&|0000\rangle+|1011\rangle+|1101\rangle+|1110\rangle,\nonumber \\ 
 |\Psi^4_9\rangle&=&|0000\rangle+|0111\rangle,
\end{eqnarray}
\end{widetext}
\normalsize 
where the complex parameters $a_1$, $a_2$, $a_3$, and $a_4$ have non-negative real parts.

\end{document}